\newcommand*{\TechReport}{}%
\newcommand*{\addspace}{}%

\newcommand*{\IEEERelatedWork}{}%

\ifdefined\JournalVer
\documentclass[10pt,journal,twocolumn]{IEEEtran}
\else

\ifdefined\TechReport
\documentclass[10pt,journal,twocolumn]{IEEEtran}
\else
\documentclass[10pt, conference, letterpaper]{IEEEtran}
\fi

\fi

\usepackage{graphicx}
\usepackage{cite}
\usepackage[cmex10]{amsmath}
\usepackage{amssymb}
\usepackage{comment}
\usepackage{float}
\usepackage{subcaption}
\ifdefined\JournalVer
\else
\usepackage[hidelinks]{hyperref}
\fi
\usepackage{url}
\usepackage{amsfonts}
\usepackage{amsthm}
\newtheorem{theorem}{Theorem}
\newtheorem{claim}{Claim}

\newtheorem{definition}{Definition}
\newtheorem{example}{Example}
\newlength{\grafflecm}
  {\begin{trivlist}\item[]{\textbf{Proof of #1:}} }%
  {\end{trivlist}}
\usepackage{clrscode}
\usepackage{array}
\usepackage{multirow}

\ifdefined\JournalVer
\makeatletter
\def\ps@headings{%
\def\@oddhead{\mbox{}\scriptsize\rightmark \hfil \thepage}%
\def\@evenhead{\scriptsize\thepage \hfil\leftmark\mbox{}}%
\def\@oddfoot{}%
\def\@evenfoot{}}
\makeatother
\fi

\makeatletter
\newenvironment{chapquote}[2][2em]
  {\setlength{\@tempdima}{#1}%
   \def\chapquote@author{#2}%
   \parshape 1 \@tempdima \dimexpr\textwidth-2\@tempdima\relax%
   \itshape}
  {\par\normalfont\hfill--\ \chapquote@author\hspace*{\@tempdima}\par\bigskip}
\makeatother

\ifdefined\BlindRev
\newcommand{\rptp}{ReversePTP}
\else
\newcommand{\rptp}{\textsc{ReversePTP}}
\fi

\newcommand{\timec}{\proc{Time4}}
\newcommand{\boldtimec}{\proc{\textbf{Time4}}}
\newcommand{\naive}{\emph{naive}}
\newcommand{\twophase}{\emph{two-phase}}
\newcommand{\order}{\emph{order}}
\newcommand{\scratch}{\emph{scratch}}

\newcommand{\add}{\textit{Add}}
\newcommand{\remove}{\textit{Remove}}
\newcommand{\offset}{\mathsf{offset}}
\newcommand{\oi}{\mathrm{\offset}_{i}}
\newcommand{\dmi}{\Delta}
\newcommand{\ivar}{I_R}
\newcommand{\schea}{scheduling error}
\newcommand{\exea}{execution accuracy}

\IEEEoverridecommandlockouts
\begin{document}
\ifdefined\JournalVer
\bstctlcite{IEEEexample:BSTcontrol}
\fi

\interfootnotelinepenalty=10000

\date{}

\title{
\ifdefined\JournalVer
\timec: Time for SDN
\else
\ifdefined\TechReport
\hspace{15mm}
\fi
\ifdefined\TechReport
\hspace{7mm}\timec: Time for SDN
\newline \newline
				\large Technical Report\textsuperscript{\ensuremath\dagger}\thanks{\textsuperscript{\ensuremath\dagger}\scriptsize This report is an extended version of~\cite{Time4Infocom}, which was accepted to IEEE INFOCOM~'16, San Francisco, April 2016. A preliminary version of this report was published in arXiv~\cite{time4ArXivOld} in May, 2015.}, February 2016
\else
Software Defined Networks: It's About Time
\fi
\fi
}

\ifdefined\BlindRev
\else
\ifdefined\JournalVer
\author{\IEEEauthorblockN{\large Tal Mizrahi, Yoram Moses}}

\else
\author{
{Tal Mizrahi, Yoram Moses\textsuperscript{\ensuremath*}\thanks{\textsuperscript{\ensuremath*}\scriptsize Yoram Moses is the Israel Pollak academic chair at Technion.}}\\
Technion --- Israel Institute of Technology\\
Email: \{dew@tx, moses@ee\}.technion.ac.il
}
\fi
\fi

\maketitle

\thispagestyle{empty}

\begin{abstract}
With the rise of Software Defined Networks (SDN), there is growing interest in dynamic and centralized traffic engineering, where decisions about forwarding paths are taken dynamically from a network-wide perspective. Frequent path reconfiguration can significantly improve the network performance, but should be handled with care, so as to minimize disruptions that may occur during network updates. 

In this paper we introduce \timec, an approach that uses accurate time to coordinate network updates. 
\ifdefined\TechReport
\timec\ is a powerful tool in softwarized environments, that can be used for various network update scenarios.
Specifically, we
\else
We
\fi
characterize a set of update scenarios called \emph{flow swaps}, for which \timec\ is the optimal update approach, yielding less packet loss than existing update approaches. We define the \emph{lossless flow allocation problem}, and formally show that in environments with frequent path allocation, scenarios that require simultaneous changes at multiple network devices are inevitable.

We present the design, implementation, and evaluation of a \timec-enabled OpenFlow prototype. 
\ifdefined\OpenSourceSoon
The prototype will soon be publicly available as open source.
\else
The prototype is publicly available as open source. 
\fi
Our work includes an extension to the OpenFlow protocol that has been adopted by the Open Networking Foundation (ONF), and is now included in OpenFlow 1.5. 
\ifdefined\TechReport
Our experimental results show the significant advantages of \timec\ compared to other network update approaches, and demonstrate an SDN use case that is infeasible without \timec.
\else
Our experimental results demonstrate the significant advantages of \timec\ compared to other network update approaches.
\fi
\end{abstract}

\ifdefined\JournalVer
\begin{IEEEkeywords}
SDN, time, clock synchronization, network updates.
\end{IEEEkeywords}

\maketitle
\IEEEdisplaynotcompsoctitleabstractindextext
\fi

\vspace{5mm}

\begin{chapquote}{\textit{Ray Cummings}}
\hspace{-5mm}
Time is what keeps everything from happening at once
\end{chapquote}

\vspace{-2mm}

\section{Introduction}
\subsection{It's About Time}

\ifdefined\JournalVer
{\let\thefootnote\relax\footnotetext{
This manuscript is an extended version of~\cite{Time4Infocom},
which was accepted to IEEE INFOCOM '16, San Francisco, April 2016.

\ifdefined\addspace \vspace{1mm} \fi
This submission includes the following new technical contributions:

$\bullet$ A new subsection has been added (Section~\ref{NetUtilSec}), presenting new theoretical analysis of how the network utilization affects flow swaps. New theoretical analysis has also been added about scaling the results to a large number of paths, including the newly added Theorem~\ref{MSwapTheo}.

$\bullet$ The current version incorporates the proofs of all theorems, including the proofs of Theorems~\ref{SwapImpactTheorem} and~\ref{NSwapTheo}, which were excluded from the conference version. 

$\bullet$ The experimental evaluation section (Section~\ref{EvaluationSec}) has been significantly extended with new experimental results, including the video swapping experiment, which is presented in a new subsection (\ref{MicrobSec}). Other experimental results that were not included in the conference paper have also been added to the current version (see Figures~\ref{fig:CDF} and~\ref{fig:LossvsDelta}).

\ifdefined\addspace \vspace{2mm} \fi

Tal Mizrahi and Yoram Moses are with the Department of Electrical Engineering, Technion, Haifa 32000,
Israel (e-mails: dew@tx.technion.ac.il, moses@ee.technion.ac.il). Yoram Moses is the Israel Pollak academic chair at Technion.

}}
\fi

The use of synchronized clocks was first introduced in the $19^{th}$ century by the Great Western Railway company in Great Britain. Clock synchronization has significantly evolved since then, and is now a mature technology that is being used by various different applications, from mobile backhaul networks~\cite{G8271} to distributed databases~\cite{corbett2013spanner}.


The Precision Time Protocol (PTP), defined in the IEEE 1588 standard~\cite{IEEE1588}, can synchronize clocks to a very high degree of accuracy, typically on the order of 1~microsecond~\cite{ChinaMobile,G8271,C37.238}. PTP is a common and affordable feature in commodity switches.
Notably, 9 out of the 13~SDN-capable switch silicons listed in the Open Networking Foundation (ONF) SDN Product Directory~\cite{ONFSDNProd} have native IEEE 1588 
\ifdefined\ShortVersion
support.
\else
support~\cite{BCM56840, BCM56850, CTC6048, FM5000, LSI, Mellanox, Tilera, Armada, HX4100}. 
\fi

In this work we introduce \timec, a \textbf{generic} tool for using time in SDN. One of the products of this work is a new feature that enables timed updates in OpenFlow, and has been incorporated in OpenFlow 1.5. Furthermore, we present a class of update scenarios in which the use of accurate time is provably optimal, while existing update methods are sub-optimal.

\subsection{The Challenge of Dynamic Traffic Engineering in SDN}
\begin{sloppypar}
Defining network routes dynamically, based on a complete view of the network, can significantly improve the network performance compared to the use of distributed routing protocols. 
SDN and OpenFlow~\cite{McKeownOpenflow,OpenFlow1.4} have been leading trends in this context, but several other ongoing efforts offer similar  
\ifdefined\ShortVersion
concepts (e.g., ~\cite{i2rs}).
\else
concepts.
The Interface to the Routing System (I2RS) working group~\cite{i2rs}, and the Forwarding and Control Element Separation (ForCES) working group~\cite{forces} are two examples of such ongoing efforts in the Internet Engineering Task Force (IETF). 
\fi
\end{sloppypar}

Centralized network
\ifdefined\TechReport
updates, whether they are related to network topology, security policy, or other configuration attributes, 
\else
updates
\fi
often involve multiple network devices. Hence, updates must be performed in a way that strives to minimize temporary anomalies such as traffic loops, congestion, or disruptions, which may occur during transient states where the network has been partially updated.

While SDN was originally considered in the context of campus networks~\cite{McKeownOpenflow} and data centers~\cite{al2010hedera}, it is now also being considered for Wide Area Networks (WANs)~\cite{hong2013achieving, jain2013b4}, carrier networks, and mobile backhaul networks~\cite{onfmobile}. 

WAN and carrier-grade networks require a very low packet loss rate. Carrier-grade performance is often associated with the term \emph{five nines}, representing an availability of 99.999\%. Mobile backhaul networks require a Frame Loss Ratio (FLR) of no more than $10^{-4}$ for voice and video traffic, and no more than $10^{-3}$ for lower priority traffic~\cite{MEF22.1}. Other types of carrier network applications, such as storage and financial trading require even lower loss rates~\cite{MEF23.1}, on the order of $10^{-5}$. 

Several recent works have explored the realm of dynamic path reconfiguration, with frequent updates on the order of minutes~\cite{hong2013achieving, jain2013b4, jin2014dynamic}, enabled by SDN. Interestingly, for voice and video traffic, a frame loss ratio of up to $10^{-4}$ implies that service must not be disrupted for more than $6$~milliseconds per minute. Hence, if path updates occur on a per-minute basis, then transient disruptions must be limited to a short period of no more than a few milliseconds.

\subsection{Timed Network Updates}
\label{TimedSec}
\ifdefined\TechReport
We explore the use of \emph{accurate time} as a tool for performing coordinated network updates in a way that minimizes packet loss. Softwarized management can significantly benefit from using time for \emph{coordinating} network-wide orchestration, and for enforcing a given \emph{order} of events.
We introduce \timec, which is an update approach that performs multiple changes at different switches at the same time. 
\else
We explore the use of \emph{accurate time} as a tool for performing coordinated network updates in a way that minimizes packet loss. We introduce \timec, which is an update approach that performs multiple changes at different switches at the same time. 
\fi

\begin{example}
Fig.~\ref{fig:Swap} illustrates a \emph{flow swapping} scenario. In this scenario, the forwarding paths of two flows, $f_1$ and $f_2$, need to be reconfigured, as illustrated in the figure. It is assumed that all links in the network have an identical capacity of 1 unit, and that both $f_1$ and $f_2$ require a bandwidth of 1 unit. 
In the presence of accurate clocks, by scheduling $S_1$ and $S_3$ to update their paths at the same time, there is no congestion during the update procedure, and the reconfiguration is smooth. As clocks will typically be reasonably well synchronized, albeit not perfectly synchronized, such a scheme will result in a very short period of congestion. 
\end{example}

\begin{figure}[htbp]
	\ifdefined\cutspace \vspace{-5mm} \fi
  \centering
  \fbox{\includegraphics[width=.47\textwidth]{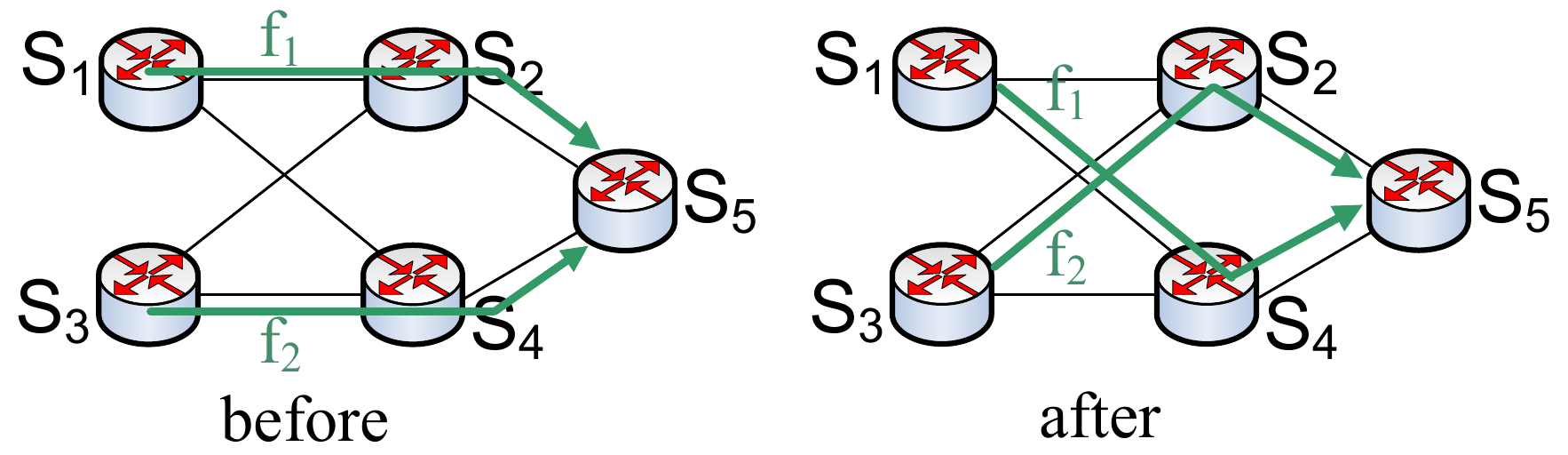}}
	\captionsetup{justification=centering}
  \caption{Flow Swapping---Flows need to convert from the ``before'' configuration to the ``after''.}
  \label{fig:Swap}
	\ifdefined\cutspace \vspace{-2mm} \fi
\end{figure}

In this paper we show that in a dynamic environment, where flows are frequently added, removed or rerouted, flow swaps are inevitable. A notable example of the importance of flow swaps is a recently published work by Fox Networks~\cite{edwards2014using}, in which accurately timed flow swaps are essential in the context of video switching.

One of our key results is that simultaneous updates are the optimal approach in scenarios such as Example~1, whereas other update approaches may yield considerable packet loss, or incur higher resource overhead. Note that such packet loss can be reduced either by increasing the capacity of the communication links, or by increasing the buffer memories in the switches. We show that for a given amount of resources, \timec\ yields lower packet loss than other approaches.

\textbf{Accuracy} is a key requirement in \timec; since updates cannot be applied at the exact same instant at all switches, they are performed within a short time interval called the \emph{\schea}. The experiments we present in Section~\ref{EvaluationSec} show that the \schea\ in \textbf{software} switches is on the order of 1~millisecond. The TCAM-based \textbf{hardware} solution of~\cite{Infocom-TimeFlip} can execute scheduled events in existing switches with an accuracy \textbf{on the order of 1~microsecond}. 



\ifdefined\TechReport
\begin{sloppypar}
Accurate time is a powerful abstraction for SDN programmers, not only for flow swaps, but also for \textbf{timed consistent updates}, as discussed by~\cite{TimedConsistent}.
\end{sloppypar}
\fi

\ifdefined\IEEERelatedWork
\subsection{Related Work}
\label{RelatedSec}
\ifdefined\TechReport
Time and synchronized clocks have been used in various distributed applications, from mobile backhaul networks~\cite{G8271} to distributed databases~\cite{corbett2013spanner}.  
\ifdefined\ShortVersion
Time-of-day routing~\cite{ash1985use} routes traffic to different destinations based on the time-of-day, but is typically performed at a low rate and does not place demanding requirements on accuracy. 
\else
Time-of-day routing~\cite{ash1985use} routes traffic to different destinations based on the time-of-day. Path calendaring~\cite{kandula2014calendaring} can be used to configure network paths based on scheduled or foreseen traffic changes. The two latter examples are typically performed at a low rate and do not place demanding requirements on accuracy. 
\fi
\fi

Various network update approaches have been analyzed in the literature. A common approach is to use a sequence of configuration commands~\cite{francois2007avoiding, vanbever2011seamless, liu2013zupdate, jin2014dynamic}, whereby the \textbf{order} of execution guarantees that no anomalies are caused in intermediate states of the procedure. However, as observed by~\cite{jin2014dynamic}, in some update scenarios, known as \textbf{deadlocks}, there is no order that guarantees a consistent transition. 
\textbf{\emph{Two-phase}} updates~\cite{reitblatt2012abstractions} use configuration version tags to guarantee consistency during updates. However, as per~\cite{reitblatt2012abstractions}, \twophase\ updates cannot guarantee congestion freedom, and are therefore not effective in flow swap scenarios, such as Fig.~\ref{fig:Swap}. Hence, in flow swap scenarios the \order\ approach and the \twophase\ approach produce the same result as the simple-minded approach, in which the controller sends the update commands as close as possible to instantaneously, and hopes for the best.

In this paper we present \timec, an update approach that is most effective in flow swaps and other deadlock~\cite{jin2014dynamic} scenarios, such as Fig.~\ref{fig:Swap}. We refer to update approaches that do not use time as \textbf{untimed} update approaches.

In SWAN~\cite{hong2013achieving}, the authors suggest that reserving unused \emph{scratch} capacity of 10-30\% on every link can allow congestion-free updates in most scenarios. The B4~\cite{jain2013b4} approach prevents packet loss during path updates by temporarily reducing the bandwidth of some or all of the flows. Our approach does not require scratch capacity, and does not reduce the bandwidth of flows during network updates. Furthermore, in this paper we show that variants of SWAN and B4 that make use of \timec\ can perform better than the original versions.


\ifdefined\TechReport
A recently published work by Fox Networks~\cite{edwards2014using} shows that accurately timed path updates are essential for video swapping.
We analyze this use case further in Section~\ref{EvaluationSec}.
\else
\fi

Rearrangeably non-blocking topologies (e.g., ~\cite{pippenger1978rearrangeable}) allow new traffic flows to be added to the network by rearranging existing flows. The analysis of flow swaps presented in this paper emphasizes the requirement to perform \emph{simultaneous} reroutes during the rearrangement procedure, an aspect which has not been previously studied. 

\ifdefined\BlindRev
The concept of using accurate time to trigger policy and routing changes was briefly discussed in~\cite{greenberg2005clean}, as well as in two work-in-progress papers~\cite{hotsdn,onstime}. The use of time for \emph{consistent} updates was discussed in~\cite{TimedConsistent}.
\else
Preliminary work-in-progress versions of the current paper introduced the concept of using time in SDN~\cite{hotsdn} and the flow swapping scenario~\cite{onstime}. The use of time for \emph{consistent} updates was discussed in~\cite{TimedConsistent}. TimeFlip~\cite{Infocom-TimeFlip} presented a practical method of implementing timed updates.
\fi
The current work is the first to present a generic protocol for performing timed updates in SDN, and the first to analyze \emph{flow swaps}, a natural application in which timed updates are the optimal update approach. 

\fi 

\subsection{Contributions}
The main contributions of this paper are as follows:
\begin{itemize}
	\item We consider a class of network update scenarios called \emph{flow swaps}, and show that simultaneous updates using synchronized clocks are provably the optimal approach of implementing them. In contrast, existing approaches for consistent updates (e.g.,~\cite{reitblatt2012abstractions,jin2014dynamic}) are not applicable to flow swaps, and other update approaches such as SWAN~\cite{hong2013achieving} and B4~\cite{jain2013b4} can perform flow swaps, but at the expense of increased resource overhead.
\ifdefined\TechReport
	\item We use game-theoretic analysis to show that flow swaps are inevitable in the dynamic nature of SDN. 
\fi
	\item We present the design, implementation and evaluation of a prototype that performs timed updates in OpenFlow. 
	\item Our work includes an extension to the OpenFlow protocol that has been approved by the ONF and integrated into OpenFlow 1.5~\cite{OpenFlow1.5}, and into the OpenFlow 1.3.x extension package~\cite{OpenFlow1.3ext}. 
\ifdefined\OpenSourceSoon
The source code of our prototype will soon be publicly available.
\else
The source code of our prototype is publicly available~\cite{TimedSDNSource}.
\fi
	\item We present experimental results that demonstrate the advantage of timed updates over existing approaches. Moreover, we show that existing update approaches (SWAN and B4) can be improved by using accurate time.
\ifdefined\TechReport
	\item Our experiments include an emulation of an SDN-controlled video swapping scenario, a real-life use case that has been shown~\cite{edwards2014using} to be infeasible with previous versions of OpenFlow, which did not include our time extension.
\fi
\end{itemize}


\ifdefined\TechReport
\else
\ifdefined\BlindRev
Due to space limits, some of the proofs and experimental results are presented in an anonymous technical report~\cite{TimeConfTR}.
\else
Due to space limits, some of the proofs and experimental results are presented in~\cite{TimeConfTR}.
\fi
\fi

\begin{sloppypar}
\section{The Lossless Flow Allocation (LFA) Problem}
\label{LFASec}
\end{sloppypar}
\subsection{Inevitable Flow Swaps}
Fig.~\ref{fig:Swap} presents a scenario in which it is necessary to \emph{swap} two flows, i.e., to update two switches at the same time. In this section we discuss the inevitability of flow swaps; we show that there does not exist a controller routing strategy that avoids the need for flow swaps. 

Our analysis is based on representing the flow-swap problem as an instance of an unsplittable flow problem, as illustrated in Fig.~\ref{fig:FlowGraph}. The topology of the graph in Fig.~\ref{fig:FlowGraph} models the traffic behavior to a given destination in common multi-rooted network topologies such as fat-tree and Clos (Fig.~\ref{fig:Flows}). 

\begin{figure}[!t]

	\centering
  \begin{subfigure}[t]{.2\textwidth}
  \centering
  \fbox{\includegraphics[height=11\grafflecm]{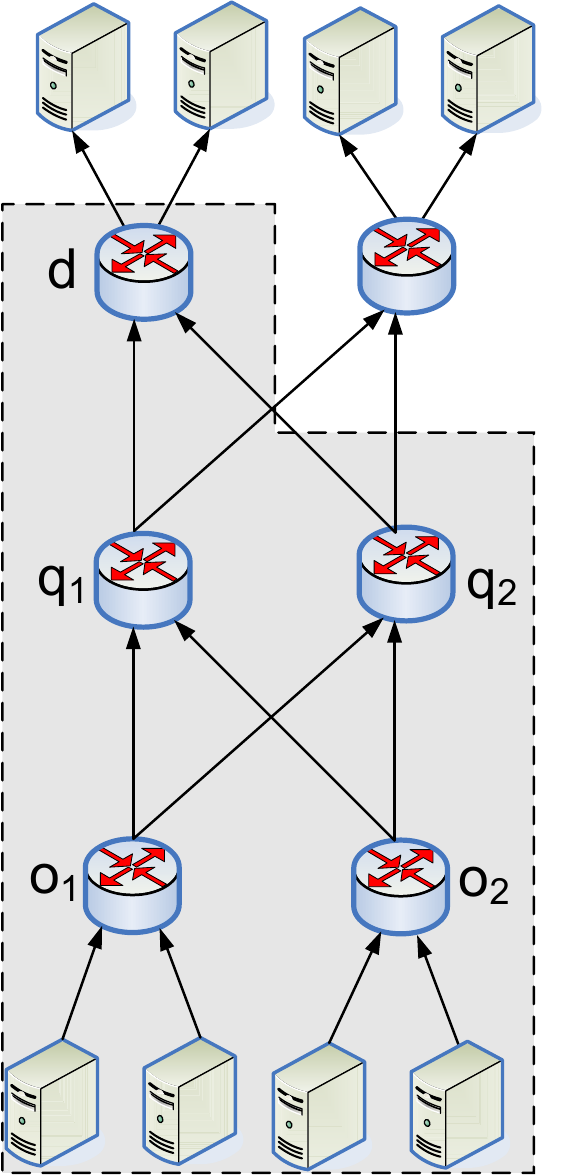}}
	\captionsetup{justification=centering}
  \caption{Clos network.}
  \label{fig:Flows}
  \end{subfigure}%
  \begin{subfigure}[t]{.3\textwidth}
	\centering
  \fbox{\includegraphics[height=11\grafflecm]{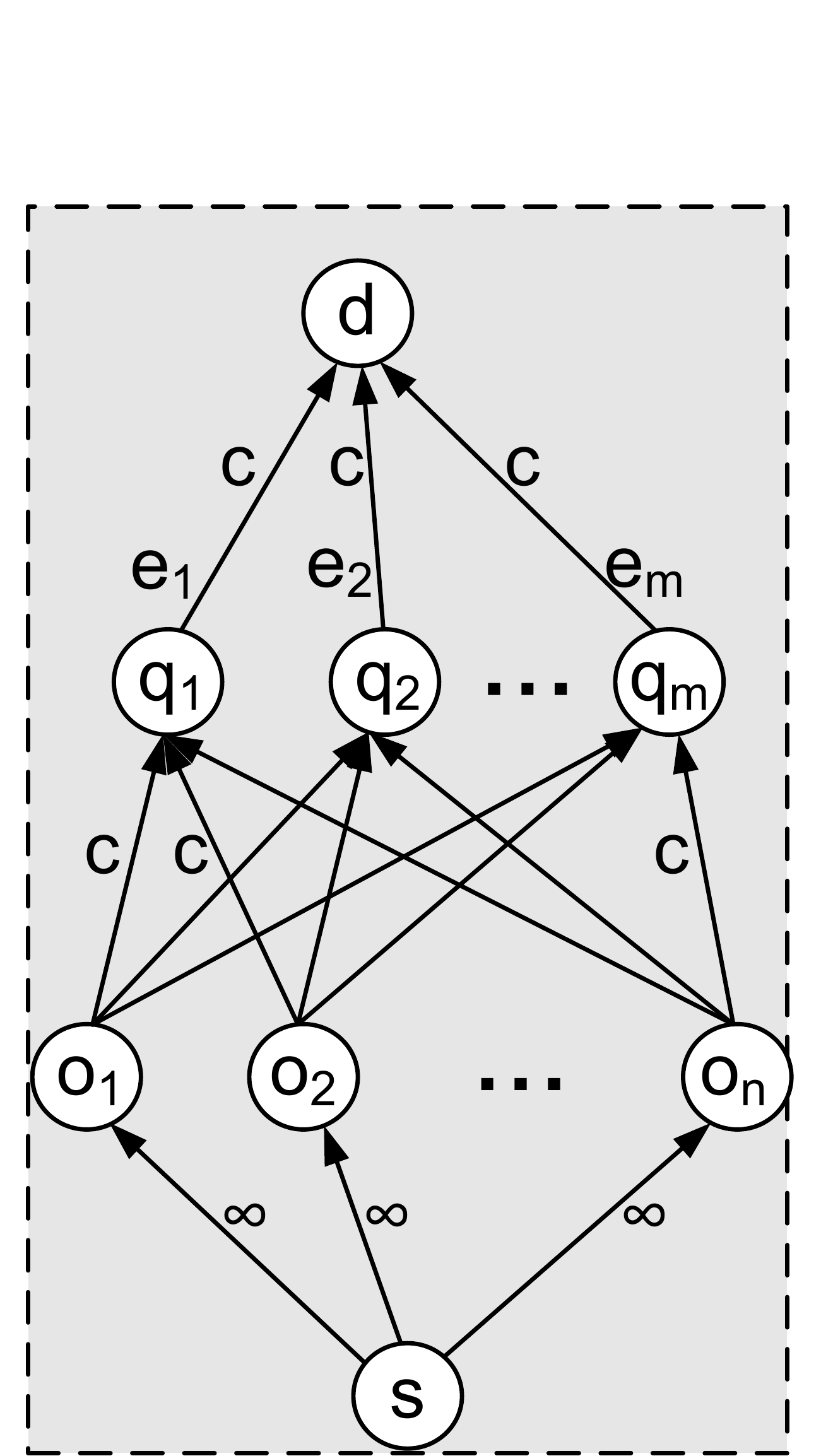}}
	\captionsetup{justification=centering}
  \caption{Unsplittable flow graph.}
  \label{fig:FlowGraph}
  \end{subfigure}%

  \caption{Modeling a Clos topology as an unsplittable flow graph.}
  \label{fig:SwapModel}
	\ifdefined\cutspace \vspace{-5mm} \fi
\end{figure}


The unsplittable flow problem~\cite{Kleinberg96} has been thoroughly discussed in the literature; given a directed graph, a source node $s$, a destination node $d$, and a set of flow demands (commodities) between $s$ and $d$, the goal is to maximize the traffic rate from the source to the destination. In this paper we define a \emph{game} between two players: a \textbf{source}\footnote{The source player does not represent a malicious attacker; it is an `adversary', representing the worst-case scenario.} that generates traffic flows (commodities) and a \textbf{controller} that reconfigures the network forwarding rules in a way that allows the network to forward all traffic generated by the source without packet losses. 


Our main argument, phrased in Theorem~\ref{SwapTheo}, is that the source has a strategy that \textbf{forces} the controller to perform a flow swap, i.e., to reconfigure the path of two or more flows at the same time. 
Thus, a scenario in which multiple flows must be updated at the \textbf{same time} is inevitable, implying the importance of timed updates. 

Moreover, we show that the controller can be forced to invoke $n$ individual commands that should optimally be performed at the same time. Update approaches that do not use time, also known as \textbf{untimed} approaches, cause the updates to be performed over a long period of time, potentially resulting in slow and possibly erratic response times and significant packet loss. Timed coordination allows us to perform the $n$ updates within a short time interval that depends on the scheduling error.

\begin{sloppypar}
Although our analysis focuses on the topology of~Fig.\ref{fig:FlowGraph}, it can be shown that the results are applicable to other topologies as well, where the source can force the controller to perform a swap over the edges of the min-cut of the graph. 
\end{sloppypar}

\subsection{Model and Definitions}
\label{ModelSec}
We now introduce the \emph{lossless flow allocation (LFA)} problem; it is not presented as an optimization problem, but rather as a game between two players: a \textbf{source} 
and a \textbf{controller}. 
As the source adds or removes flows (commodities), the controller reconfigures the forwarding rules so as to guarantee that all flows are forwarded without packet loss. \textbf{The controller's goal} is to find a forwarding path for all the flows in the system without exceeding the capacity of any of the edges, i.e., to completely avoid loss of packets from the given flows. \textbf{The source's goal} is to progressively add flows, without exceeding the network's capacity, forcing the controller to perform a flow swap.
We shall show that the source has a strategy that forces the controller to swap traffic flows simultaneously in order to avoid packet loss.

Our model makes three basic assumptions: (i) each flow has a \textbf{fixed bandwidth}, (ii) the controller strives to \textbf{avoid packet loss}, and (iii) flows are \textbf{unsplittable}. We discuss these assumptions further in Sec.~\ref{DiscussionSec}.

\ifdefined\TechReport
The term \emph{flow} in classic flow problems typically refers to the amount of traffic that is forwarded through each edge of the graph. Since our analysis focuses on SDN, we slightly divert from the common flow problem terminology, and use the term \emph{flow} in its OpenFlow sense, i.e., a set of packets that share common properties, such as source and destination network addresses. A flow in our context, can be seen as a session between the source and destination that runs traffic at a fixed rate.
\else
We use the term \emph{flow} in its OpenFlow sense, i.e., a set of packets that share common properties, such as source and destination network addresses. A flow in our context, can be seen as a session between the source and destination that runs traffic at a fixed rate.
\fi

The network is represented by a directed weighted acyclic graph (Fig.~\ref{fig:FlowGraph}), $G=(\mathbb{V},E,c)$, with a source $s$, a destination $d$, and a set of intermediate nodes, $\mathbb{V}_{in}$. Thus, $\mathbb{V}=\mathbb{V}_{in}\cup \{s,d\}$. The nodes directly connected to $s$ are denoted by $\mathbb{O}=\{o_1, o_2, \ldots, o_n\}$. Each of the outgoing edges from the source $s$ has an infinite capacity, whereas the rest of the edges have a capacity $c$. For the sake of simplicity, and without loss of generality, throughout this section we assume that $c=1$. Such a graph $G$ is referred to as an \emph{LFA graph}.

The source node progressively transmits traffic flows towards the destination node. Each flow represents a session between $s$ and $d$; every flow has a constant bandwidth, and cannot be split between two paths. A centralized controller configures the forwarding policy of the intermediate nodes, determining the path of each flow. Given a set of flows from $s$ to $d$, the controller's goal is to configure the forwarding policy of the nodes in a way that allows all flows to be forwarded to~$d$ without exceeding the capacity of any of the edges.

The set of flows that are generated by $s$ is denoted by $\mathbb{F} ::= \{ F_1, F_2, \ldots, F_k \}$. Each flow $F_i$ is defined as $F_i ::= (i, f_i, r_i)$, where $i$ is a unique flow index, $f_i$ is the bandwidth satisfying $0 < f_i \leq c$, and $r_i$ denotes the node that the controller forwards the flow to, i.e., $r_i \in  \{o_1, o_2, \ldots, o_n\}$. 

It is assumed that the controller monitors the network, and thus it is aware of the flow set $\mathbb{F}$. The controller maintains a forwarding function, $R_{con} : \mathbb{F} \times \mathbb{V}_{in} \longrightarrow \mathbb{V}_{in}\cup \{d\}$. Every node (switch) has a flow table, consisting of a set of \emph{entries}; an element $w \in \mathbb{F} \times \mathbb{V}_{in}$ is referred to as an \emph{entry} for short. An update of $R_{con}$ is defined to be a partial function $u : \mathbb{F} \times \mathbb{V}_{in} \rightharpoonup \mathbb{V}_{in}\cup \{d\}$.
We define a \emph{reroute} as an update $u$ that has a single entry in its domain. We call an update that has more than one entry in its domain a \emph{swap}, and it is assumed that all updates in a \emph{swap} are performed at the same time. We define a $k$-swap for $k \geq 2$ as a swap that updates entries in at least $k$ different nodes. Note that a $k$-swap is possible only if $n \geq k$, where $n$ is the number of nodes in $\mathbb{O}$.
We focus our analysis on $2$-swaps, and throughout the section we assume that $n \geq 2$. In Section~\ref{NSwapSec} we discuss $k$-swaps for values of $k>2$.

\subsection{The LFA Game}
The lossless flow allocation problem can be viewed as a game between two players, the source and the controller. The game proceeds by a sequence of steps; in each step the source either adds or removes a single flow (Fig.~\ref{fig:SrcProc}), and then waits for the controller to perform a sequence of updates (Fig.~\ref{fig:CtlProc}). The source's strategy $\mathbb{S}_s(\mathbb{F}, R_{con})=(a,F)$, is a function that defines for each flow set $\mathbb{F}$ and forwarding function $R_{con}$ for $\mathbb{F}$, a pair $(a,F)$ representing the source's next step, where $a \in \{\add, \remove \}$ is the action to be taken by the source, and $F=(j,f_j,r_j)$ is a single flow to be added or removed. The controller's strategy is defined by $\mathbb{S}_{con}(R_{con}, a, F)=\mathbb{U}$, where $\mathbb{U}=\{u_1,\ldots,u_\ell\}$ is a sequence of updates, such that (i) at the end of each update no edge exceeds its capacity, and (ii) at the end of the last update, $u_\ell$, the forwarding function $R_{con}$ defines a forwarding path for all flows in $\mathbb{F}$. Notice that when a flow is to be removed, the controller's update is trivial; it simply removes all the relevant entries from the domain of $R_{con}$. Hence our analysis focuses on \emph{adding} new flows.

\ifdefined\JournalVer
\begin{figure}[!t]
\else
\begin{figure}[!b]
\fi
\ifdefined\cutspace \vspace{-3mm} \fi
\hrule
\ifdefined\cutspace \vspace{-2mm} \fi
  \begin{codebox}
    \Procname{$\proc{Source Procedure}$}
		\li $\mathbb{F} \gets \emptyset$
		
		\li \textbf{repeat} at every step
		\li \ \ $(a,F) \gets \mathbb{S}_s(\mathbb{F}, R_{con})$
		\li \ \ \textbf{if} $a = \add$
		\li \ \ \ \ $\mathbb{F} \gets \mathbb{F} \cup F$
		\li \ \ \ \ Wait for the controller to complete updates
		\li \ \ \textbf{else} // $a = \remove$
		\li \ \ \ \ $\mathbb{F} \gets \mathbb{F} \setminus F$
  \end{codebox}
	\ifdefined\cutspace \vspace{-2mm} \fi
  \hrule
  \caption{The LFA game: the source's procedure.}
	\ifdefined\cutspace \vspace{-3mm} \fi
  \label{fig:SrcProc}
\end{figure}

\ifdefined\JournalVer
\begin{figure}[!h]
\else
\begin{figure}[!t]
\fi
\hrule
\ifdefined\cutspace \vspace{-2mm} \fi
  \begin{codebox}
    \Procname{$\proc{Controller Procedure}$}
		
		\li \textbf{repeat} at every step
		\li \ \ $\{u_1, \ldots,u_\ell \} \gets \mathbb{S}_{con}(R_{con}, a, F)$
		\li \ \ \textbf{for} $j \in [1,\ell]$
		\li \ \ \ \ Update $R_{con}$ according to $u_j$
  \end{codebox}
  \ifdefined\cutspace \vspace{-2mm} \fi
  \hrule
  \caption{The LFA game: the controller's procedure.}
	\ifdefined\cutspace \vspace{-7mm} \fi
  \label{fig:CtlProc}
\end{figure}

The following theorem, which is the crux of this section, argues that the source has a strategy that forces the controller to perform a swap, and thus that flow swaps are inevitable from the controller's perspective.

\begin{theorem}
\label{SwapTheo}
Let $G$ be an LFA graph. In the LFA game over $G$, there exists a strategy, $\mathbb{S}_s$, for the source that forces every controller strategy, $\mathbb{S}_{con}$, to perform a $2$-swap.
\end{theorem}

\begin{proof}
Let $m$ be the number of incoming edges to the destination node $d$ in the LFA graph (see Fig~\ref{fig:FlowGraph}).
For $m=1$ the claim is trivial. Hence, we start by proving the claim for $m=2$, i.e., there are two edges connected to node $d$, edges $e_1$ and $e_2$.
We show that the source has a strategy that, regardless of the controller's strategy, forces the controller to use a swap. 
In the first four steps of the game, the source generates four flows, $F_1=(1,0.35,o_1)$, $F_2=(2,0.35,o_1)$, $F_3=(3,0.45,o_2)$, and $F_4=(4,0.45,o_2)$, respectively. According to the Source Procedure of Fig.~\ref{fig:SrcProc}, after each flow is added, the source waits for the controller to update $R_{con}$ before adding the next flow.
After the flows are added, there are two possible cases:
\begin{itemize}
	\item[(a)] The controller routes symmetrically through $e_1$ and $e_2$, i.e. a flow of $0.35$ and a flow of $0.45$ through each of the edges. In this case the source's strategy at this point is to generate a new flow $F_5=(5,0.3,o_1)$ with a bandwidth of $0.3$. The only way the controller can accommodate $F_5$ is by routing $F_1$ and $F_2$ through the same edge, allowing the new $0.3$ flow to be forwarded through that edge. Since there is no sequence of \emph{reroute} updates that allows the controller to reach the desired $R_{con}$, the only way to reach a state where $F_1$ and $F_2$ are routed through the same edge is to swap a $0.35$ flow with a $0.45$ flow. Thus, by issuing $F_5$ the controller forces a flow swap as claimed.
	\item[(b)] The controller routes $F_1$ and $F_2$ through one edge, and $F_3$ and $F_4$ through the other edge. In this case the source's strategy is to generate two flows, $F_6$ and $F_7$, with a bandwidth of $0.2$ each. The controller must route $F_6$ through the edge with $F_1$ and $F_2$. Now each path sustains a bandwidth of $0.9$ units. Thus, when $F_7$ is added by the source, the controller is forced to perform a swap between one of the $0.35$ flows and one of the $0.45$ flows. 
\end{itemize}
In both cases the controller is forced to 
perform a $2$-swap, swapping a flow from~$o_1$ with a flow from~$o_2$. This proves the claim for $m=2$.

The case of $m>2$ is obtained by reduction to $m=2$: the source first generates $m-2$ flows with a bandwidth of~$1$ each, causing the controller to saturate $m-2$ edges connected to node $d$ (without loss of generality $e_3, \ldots, e_m$). At this point there are only two available edges, $e_1$ and $e_2$. From this point, the proof is identical to the case of $m=2$.
\end{proof}

The proof of Theorem~\ref{SwapTheo} showed that the controller can be forced to perform a flow swap that involves $m=2$ paths. For $m>2$, we assumed that the source saturates $m-2$ paths, reducing the analysis to the case of $m=2$. In the following theorem we show that for $m>2$ the controller can be forced to perform $\lfloor \frac{m}{2} \rfloor$ swaps.

\begin{theorem}
\label{MSwapTheo}
Let $G$ be an LFA graph. In the LFA game over $G$, if $m>2$ then there exists a strategy, $\mathbb{S}_s$, for the source that forces every controller strategy, $\mathbb{S}_{con}$, to perform $\lfloor \frac{m}{2} \rfloor$ $2$-swaps.
\end{theorem}

\begin{proof}
Assume that $m$ is even. The source generates $m$ flows with a bandwidth of $0.35$, $m$ flows with a bandwidth of $0.45$, and $m$ flows with a bandwidth of $0.2$. 
The only way the controller can route these flows without packet loss is as follows: each path sustains three flows with three different bandwidths, $0.2$, $0.35$, and $0.45$.
Now the source removes the $m$ flows of $0.2$, and adds $\frac{m}{2}$ flows of $0.3$. As in case (a) of the proof of Theorem~\ref{SwapTheo}, adding each flow of $0.3$ causes a $2$-swap. The controller is thus is forced to perform $\frac{m}{2}=\lfloor \frac{m}{2} \rfloor$ swaps.

If $m$ is odd, then the source can saturate one of the edges by generating a flow with a bandwidth of $1$, and then repeat the procedure above for the remaining $m-1$ edges, yielding $\frac{m-1}{2}=\lfloor \frac{m}{2} \rfloor$ swaps.
\end{proof}

For simplicity, throughout the rest of this section we assume that $m=2$. However, as in Theorem~\ref{MSwapTheo}, the analysis can be extended to the case of $m>2$.

\subsection{The Impact of Flow Swaps}
We define a \textbf{metric} for flow swaps, by considering the oversubscription that is caused if the flows are \textbf{not} swapped simultaneously, but updated using an untimed approach.

\ifdefined\TechReport
\else
\begin{figure*}[t]
  \centering
  \fbox{\includegraphics[width=.8\textwidth]{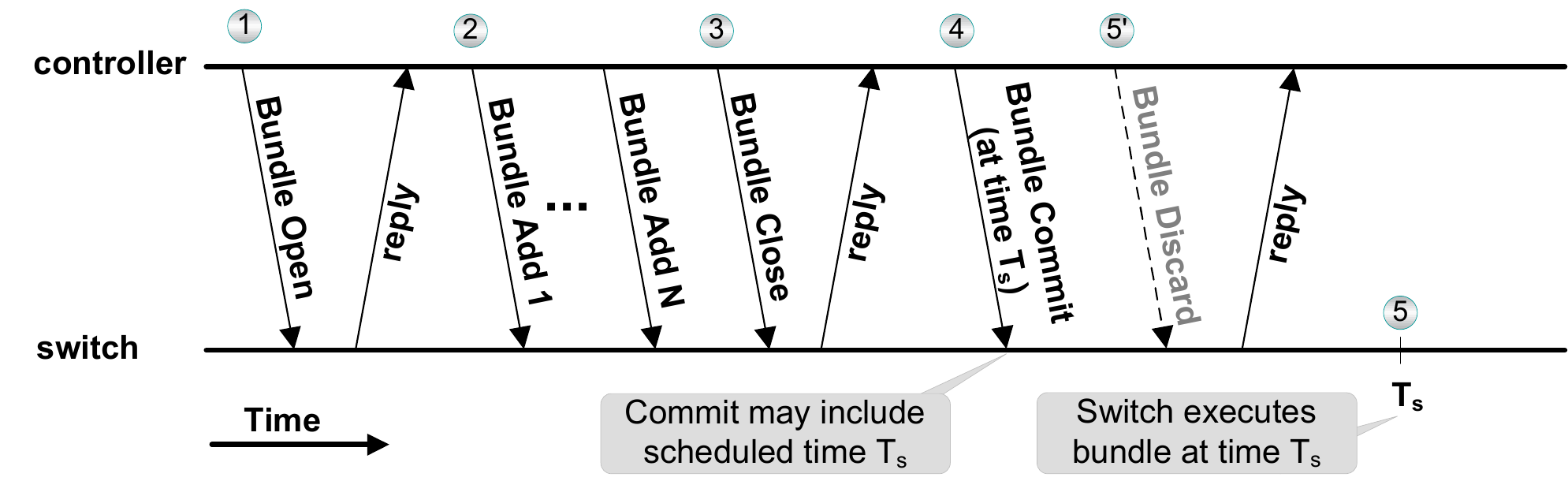}}
	\captionsetup{justification=raggedright}
  \caption{A \emph{Scheduled Bundle}: the \emph{Bundle Commit} message may include $T_s$, the scheduled time of execution. The controller can use a \emph{Bundle Discard} message to cancel the \emph{Scheduled Bundle} before $T_s$.}
  \label{fig:Bundle}
\end{figure*}
\fi

We define the \emph{oversubscription} of an edge, $e$, with respect to a forwarding function, $R_{con}$, to be the difference between the total bandwidth of the flows forwarded through~$e$ according to $R_{con}$, and the capacity of $e$. If the total bandwidth of the flows through~$e$ is less than the capacity of $e$, the oversubscription is defined to be zero.

\begin{definition}[Flow swap impact]
Let $\mathbb{F}$ be a flow set, and $R_{con}$ be the corresponding forwarding function. Consider a $2$-swap $u : \mathbb{F} \times \mathbb{V} \rightharpoonup \mathbb{V}\cup \{d\}$, such that $u=u_1 \cup u_2$, where $u_i=(w_i, v_i)$, for $w_i \in \mathbb{F} \times \mathbb{V}$, $v_i \in \mathbb{V}\cup\{d\}$, and $i \in \{1,2\}$. The \emph{impact} of $u$ is defined to be the minimum of: (i)~The oversubscription caused by applying $u_1$ to $R_{con}$, or (ii) the oversubscription caused by applying $u_2$ to $R_{con}$.
\end{definition}

\ifdefined\TechReport
\begin{example}
We observe the scenario described in the proof of Theorem~\ref{SwapTheo}, and consider  what would happen if the two flows had not been swapped simultaneously. The scenario had two cases;
in the first case, the bandwidth through each edge is $0.8$ before the controller swaps a $0.35$ flow with a $0.45$ flow. Thus, if the $0.35$ flow is rerouted and then the $0.45$ flow, the total bandwidth through the congested edge is $0.8+0.35=1.15$, creating a temporary oversubscription of $0.15$. Thus, the flow swap impact in the first case is $0.15$. In the second case, one edge sustains a bandwidth of $0.7$, and the other a bandwidth of $0.9$. The controller needs to swap a $0.35$ flow with a $0.45$ flow. If the controller first reroutes the $0.45$ flow, then during the intermediate transition period, the congested edge sustains a bandwidth of $0.7+0.45=1.15$, and thus it is oversubscribed by $0.15$. Hence, the impact in the second case is also $0.15$. 
\end{example}
\fi

The following theorem shows that in the LFA game, the source can force the controller to perform a flow swap with a swap impact of roughly $0.5$. 

\begin{theorem}
\label{SwapImpactTheorem}
Let $G$ be an LFA graph, and let $0 < \alpha < 0.5$. In the LFA game over $G$, there exists a strategy, $\mathbb{S}_s$, for the source that forces every controller strategy, $\mathbb{S}_{con}$, to perform a swap with an impact of $\alpha$.
\end{theorem}

\ifdefined\TechReport
\begin{proof}
Let $\epsilon = 0.1-0.2 \cdot \alpha$.
We use the source's strategy from the proof of Theorem~\ref{SwapTheo}, with the exception that the bandwidths $f_1, \ldots, f_7$ of flows $F_1, \ldots, F_7$ are: $f_1=f_2=0.5-2 \epsilon$, $f_3=f_4=0.5-\epsilon$, $f_5=4 \epsilon$, and $f_6=f_7=3 \epsilon$.

As in the proof of Theorem~\ref{SwapTheo}, there are two possible cases. In case (a), the controller routes symmetrically through the two paths, utilizing $1-3 \epsilon$ of the bandwidth of each path. The source adds $F_5$ in response. To accommodate $F_5$ the controller swaps $F_1$ and $F_3$. We determine the impact of this swap by considering the oversubscription of performing an untimed update; the controller first reroutes $F_1$, and only then reroutes $F_3$. Hence, the temporary oversubscription is $1 - 3 \epsilon + 0.5-2 \epsilon - 1 = 1.5 - 5 \epsilon - 1$. Thus, the impact is $0.5-5 \epsilon = \alpha$. In case (b), the controller forwards $F_1$ through the same path as $F_2$, and $F_3$ through the same path as $F_4$. The source responds by generating $F_6$ and $F_7$. Again, the controller is forced to swap between $F_1$ and $F_3$. We compute the impact by considering an untimed update, where the controller reroutes $F_3$ first, causing an oversubscription of $1 - 4 \epsilon + 0.5 - \epsilon -1 = 0.5 - 5 \epsilon = \alpha$.
In both cases the source inflicts a flow swap with an impact of $\alpha$.
\end{proof}
\else
The proof is presented in~\cite{TimeConfTR}.
\fi

Intuitively, Theorem~\ref{SwapImpactTheorem} shows that not only are flow swaps inevitable, but they have a high impact on the network, as they can cause links to be congested by roughly $50\%$ beyond their capacity.

\ifdefined\ShortVersion
\else

\subsection{Network Utilization}
\label{NetUtilSec}
Theorem~\ref{SwapTheo} demonstrates that regardless of the controller's policy, flow swaps cannot be prevented. However, the proof of Theorem~\ref{SwapTheo} uses a scenario in which the edges leading to node $d$ are almost fully utilized, suggesting that perhaps flow swaps are inevitable only when the traffic bandwidth is nearly equal to the max-flow of the graph. Arguably, as suggested in~\cite{hong2013achieving}, by reserving some scratch capacity $\nu \cdot c$ through each of the edges, for $0 < \nu < 1$, it may be possible to avoid flow swaps. In the next theorem we show that if $\nu < \frac{1}{3}$, then flow swaps are inevitable.

\begin{theorem}
\label{ScratchTheorem}
Let $G$ be an LFA graph, in which a scratch capacity of $\nu$ is reserved on each of the edges $e_1,\ldots,e_m$, and let $\nu < \frac{1}{3}$. In the LFA game over $G$, there exists a strategy for the source, $\mathbb{S}_s$, that forces every controller strategy, $\mathbb{S}_{con}$, to perform a swap.
\end{theorem}

\ifdefined\TechReport
\begin{proof}
We consider a graph $G'$, in which the capacity of each of the edges $e_1, \ldots, e_m$ is $1-\nu$. By Theorem~\ref{SwapImpactTheorem}, for every $0 < \alpha < 0.5$, there exists a strategy for the source that forces a flow swap with an impact of $\alpha$. Thus, there exists a strategy that forces at least one of the edges to sustain a bandwidth of $\alpha \cdot (1-\nu)$. Since $\nu < \frac{1}{3}$, we have $(1-\nu)>\frac{2}{3}$, and thus there exists an $\alpha < 0.5$ such that $\alpha \cdot (1-\nu) > 1$. It follows that in the original graph $G$, with scratch capacity $\nu$, there exists a strategy for the source that forces the controller to perform a flow swap in order to avoid the oversubscribed bandwidth of $\alpha \cdot (1-\nu) > 1$.
\end{proof}
\else
The proof is presented in~\cite{TimeConfTR}.
\fi



The analysis of~\cite{hong2013achieving} showed that a scratch capacity of~10\% is enough to address the reconfiguration scenarios that were considered in that work. Theorem~\ref{ScratchTheorem} shows that even a scratch capacity of $33 \frac{1}{3}$\% does not suffice to prevent flow swaps scenarios. It follows that the 10\% reserve that~\cite{hong2013achieving} suggest may not be sufficient in general for lossless reconfiguration. 

\fi

\ifdefined\TechReport
\begin{figure*}[htbp]
  \centering
  \fbox{\includegraphics[width=.8\textwidth]{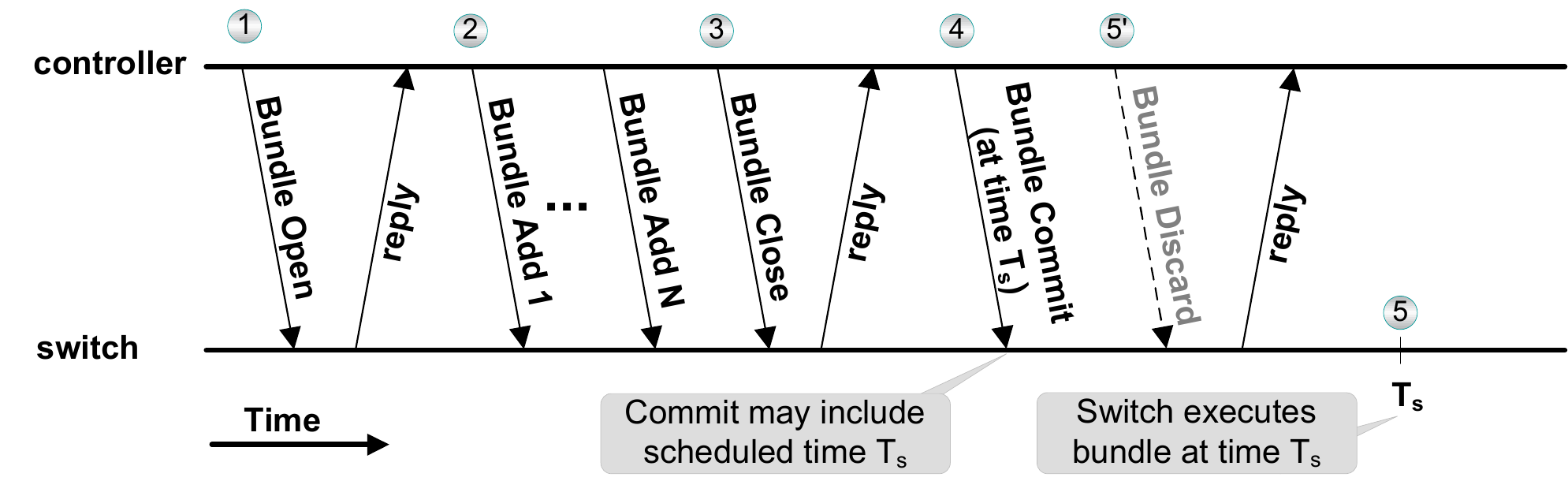}}
	\captionsetup{justification=raggedright}
  \caption{A \emph{Scheduled Bundle}: the \emph{Bundle Commit} message may include $T_s$, the scheduled time of execution. The controller can use a \emph{Bundle Discard} message to cancel the \emph{Scheduled Bundle} before time $T_s$.}
  \label{fig:Bundle}
\end{figure*}
\fi

\subsection{n-Swaps}
\label{NSwapSec}
As defined above, a $k$-swap is a swap that involves $k$ or more nodes. In previous subsections we discussed $2$-swaps. The following theorem generalizes Theorem~\ref{SwapTheo} to $n$-swaps, where $n$ is the number of nodes in $\mathbb{O}$.

\begin{theorem}
\label{NSwapTheo}
Let $G$ be an LFA graph. In the LFA game over $G$, there exists a strategy, $\mathbb{S}_s$, for the source that forces every controller strategy, $\mathbb{S}_{con}$, to perform an $n$-swap.
\end{theorem}

\ifdefined\TechReport
\begin{proof}
For $n=1$, the claim is trivial. For $n=2$, the claim was proven in Theorem~\ref{SwapTheo}. Thus, we assume $n \geq 3$.

If $m>2$, the source first generates $m-2$ flows with a rate~$c$ each, and we assume without loss of generality that after the controller allocates these flows only $e_1$ and $e_2$ remain unused. Thus, we focus on the case where $m=2$.

We describe a strategy, $\mathbb{S}_s$ as required; $s$ generates three types of flows:

\begin{itemize}
	\item Type A: two flows $F_1, F_2$, at a rate of $h$ each: $F_1=(1,h,o_1)$, and $F_2=(2,h,o_1)$.
	\item Type B: $n$ flows, $F_3, \ldots, F_{n+2}$, with a total rate $g$, i.e., at a rate of $\frac{g}{n}$ each. The source sends each of the $n$ flows through a different node of $\mathbb{O}$.
	\item Type C: $n-1$ flows, $F_{n+3}, \ldots, F_{2n+1}$ with a total rate $g$, i.e.,  $\frac{g}{n-1}$ each. The source sends each of the $n-1$ flows through a different node of $o_2, \ldots, o_n$.
\end{itemize}

We define $h$ and $g$ such that:

\begin{equation}
\label{eq:hg1}
\frac{1}{3} < h < g < \frac{1}{2}
\end{equation}

\begin{equation}
\label{eq:hg2}
g > (n^2-n)(1-2h)
\end{equation}

We claim that for every $n$ there exist $g$ and $h$ that satisfy (\ref{eq:hg1}) and (\ref{eq:hg2}). We prove this claim by finding $g$ and $h$ that satisfy the two conditions. We choose an arbitrary $g$ in the range $(\frac{11}{24},\frac{1}{2})$. We find a valid $h$ by solving $g>(n^2-n)(1-2h)$. The latter yields $h>\frac{1}{2}- \frac{\alpha}{2(n^2-n)}$. Since $n\geq 3$, we have $n^2-n\geq 6$, and thus $\frac{g}{2(n^2-n)}<\frac{0.5}{2\times 6}=\frac{1}{24}$. Clearly, $\frac{g}{2(n^2-n)}>0$. It follows that every $h$ that satisfies $\frac{1}{2}- \frac{1}{24} < h < \frac{1}{2}-0$, also satisfies $h> \frac{1}{3}$. Hence, every $g$ and $h$ in the range $(\frac{11}{24},\frac{1}{2})$ that satisfy $h<g$, also satisfy (\ref{eq:hg1}) and (\ref{eq:hg2}).

Intuitively, for $h$ and $g$ sufficiently close to $\frac{1}{2}$ (but less than $\frac{1}{2}$) (\ref{eq:hg1}) and (\ref{eq:hg2}) are satisfied.
\\

We now prove that after generating the flows $F_1,\ldots,F_{2n+1}$, the function $R_{con}$ forwards all type B flows through the same path, and all type C flows through the same path.
Assume by way of contradiction that there is a forwarding function $R_{con}$ that forwards flows $F_1,\ldots,F_{2n+1}$ without loss, but does not comply to the latter claim. We consider two distinct cases: either the two type A flows are forwarded through the same edge, or they are forwarded through two different edges.
\begin{itemize}
	\item If the two type A flows are forwarded through two different paths, then we assume that $F_1$ and the $n$ type B flows are forwarded through $e_1$ and that $F_2$ and the $n-1$ type C flows are forwarded through $e_2$. Thus, at this point each of the two edges sustains traffic at a rate of $g+h$. By the assumption, there exists an update that swaps $i < n$ flows of type B with $j<n-1$ flows of type C, such that after the swap none of the edges exceeds its capacity. Thus, the update adds the bandwidth $|j\cdot \frac{g}{n-1} - i \cdot \frac{g}{n}|$ to one of the edges, and this additional bandwidth must fit into the available bandwidth before the update, $1-g-h$. Hence, $|j\cdot \frac{g}{n-1} - i \cdot \frac{g}{n}|<c-g-h$. Note that $1-g-h<1-2h<\frac{g}{n-1}- \frac{g}{n}$, following (\ref{eq:hg1}) and (\ref{eq:hg2}). Thus we get $|j\cdot \frac{g}{n-1} - i \cdot \frac{g}{n}|<\frac{g}{n-1}- \frac{g}{n}$. It follows that $|j \cdot n - i\cdot n + i| < 1$. Since $j, i, n$ are integers, we get that $j \cdot n - i\cdot n + i = 0$, and thus $j = i \cdot \frac{n-1}{n}$. Now since $i\leq n$ and $j\leq n-1$ are both natural numbers, the only solution is $j=n-1$ and $i=n$, which means that the flows from type B are all forwarded through the same path, as well as the flows of type C, contradicting the assumption.
	\item If the two type A flows are forwarded through the same edge, their total bandwidth is $2h$, and thus the remaining bandwidth through this edge is $1-2h$. From (\ref{eq:hg2}) we have $\frac{g}{n-1}-\frac{g}{n}  > 1-2h$. We note that (i) $\frac{g}{n-1} > \frac{g}{n-1}-\frac{g}{n}$, and (ii)  $\frac{g}{n} > \frac{g}{n-1}-\frac{g}{n}$. It follows that $\frac{g}{n-1}>1-2h$, and also $\frac{g}{n}>1-2h$, and thus none of the type B or type C flows fit on the same path with $F_1$ and $F_2$. Thus, all the type B and type C flows are on the same path, contradicting the assumption. 
\end{itemize}

We have shown that all flows of type B, denoted by $\mathbb{F}^B$, must be forwarded through the same path, and that all flows of type C, denoted by $\mathbb{F}^C$, are forwarded through the same path. Thus, after the source generates the $2\cdot n +1$ flows, there are two possible scenarios:
\begin{itemize}
	\item The two type A flows are forwarded through the same path, and the type B and type C flows are forwarded through the other path. In this case $s$ generates two flows at a rate of $1-h-g$ each. To accommodate both flows the controller must swap the flows of $\mathbb{F}^B$ with $F_1$ or the flows of $\mathbb{F}^C$ with $F_2$. Both possible swaps involve $n$ entries, and thus the controller is force to perform an $n$-swap.
	\item One path is used for $F_1$ and the flows of $\mathbb{F}^C$, and the other path is used for $F_2$ and the flows of $\mathbb{F}^B$. In this case the source generates a flow with a bandwidth of $1-2h$, again forcing the controller to swap the flows of $\mathbb{F}^B$ with $F_1$ or the flows of $\mathbb{F}^C$ with $F_2$.
\end{itemize}
In both cases the controller is forced to perform a swap that involves the $n$ nodes, i.e., an $n$-swap.
\end{proof}
\else
The proof is presented in~\cite{TimeConfTR}.
\fi

\section{Design and Implementation}
\label{DesImpSec}
\subsection{Protocol Design}
\ifdefined\ShortVersion
\else
\vspace{1mm}
\emph{1) Overview}
\vspace{1mm}

A \timec-enabled system is comprised of two main components:

\begin{itemize}
	\item \textbf{OpenFlow time extension.} \timec\ is built upon the OpenFlow protocol. We define an extension to the OpenFlow protocol that enables timed updates; the controller can attach an \emph{execution time} to every OpenFlow command it sends to a switch, defining when the switch should perform the required command.  
\ifdefined\BlindRev
\else
It should be noted that the \timec\ approach is not limited to OpenFlow; we have defined a similar time extension to the NETCONF protocol~\cite{TimeNetconf}, but in this paper we focus on \timec\ in the context of OpenFlow, as described in the next subsection.
\fi
	\item \textbf{Clock synchronization.} \timec\ requires the switches and controller to maintain a local clock, allowing time-triggered events. Hence, the local clocks should be synchronized. The OpenFlow time extension we defined does not mandate a specific synchronization method. Various mechanisms may be used, e.g., the Network Time Protocol (NTP), the Precision Time Protocol (PTP)~\cite{IEEE1588}, or GPS-based synchronization. The prototype we designed and implemented uses \rptp~\cite{ispcsrptp}, as described below.
\end{itemize}

\vspace{1mm}
\emph{2) OpenFlow Time Extension}
\vspace{1mm}

\fi 

We present an extension that allows OpenFlow controllers to signal the time of execution of a command to the switches. This extension is described in full in 
\ifdefined\TechReport
\ifdefined\JournalVer
\cite{TimeConfTR}.%
\else
Appendix~\ref{ExtAppendix}.%
\fi
\else
\cite{TimeConfTR}.%
\fi
\ifdefined\BlindRev
\else
\footnote{A preliminary version of this extension was presented in~\cite{TimeTR}.}
\fi

Our extension makes use of the OpenFlow~\cite{OpenFlow1.4} \textbf{Bundle} feature; a Bundle is a sequence of OpenFlow messages from the controller that is applied as a single operation. Our time extension defines \textbf{\emph{Scheduled Bundles}}, allowing all commands of a Bundle to come into effect at a pre-determined time. This is a generic means to extend all OpenFlow commands with the scheduling feature.

Using Bundle messages for implementing \timec\ has two significant advantages: (i)~It is a generic method to add the time extension to all OpenFlow commands without changing the format of all OpenFlow messages; only the format of Bundle messages is modified relative to the Bundle message format in~\cite{OpenFlow1.4}, optionally incorporating an execution time. (ii)~The Scheduled Bundle allows a relatively straightforward way to \emph{cancel} scheduled commands, as described below.

Fig.~\ref{fig:Bundle} illustrates the \emph{Scheduled Bundle} message procedure. In step 1, the controller sends a \emph{Bundle Open} message to the switch, followed by one or more Add messages (step 2). Every \emph{Add} message encapsulates an OpenFlow message, e.g., a \emph{FLOW\_MOD} message. A \emph{Bundle Close} is sent in step 3, followed by the \emph{Bundle Commit} (step 4), which optionally includes the scheduled time of execution, $T_s$. The switch then executes the desired command(s) at time $T_s$.

The \emph{Bundle Discard} message (step 5$'$) allows the controller to enforce an all-or-none scheduled update; after the \emph{Bundle Commit} is sent, if one of the switches sends an \emph{error} message, indicating that it is unable to schedule the current Bundle, the controller can send a Discard message to all switches, canceling the scheduled operation. Hence, when a switch receives a scheduled commit, to be executed at time $T_s$, the switch can verify that it can dedicate the required resources to execute the command as close as possible to $T_s$. If the switch's resources are not available, for example due to another command that is scheduled to $T_s$, then the switch replies with an error message, aborting the scheduled commit. Significantly, this mechanism allows switches to execute the command with a guaranteed scheduling accuracy, avoiding the high variation that occurs when untimed updates are used.

The OpenFlow time extension also defines \emph{Bundle Feature Request} messages, which allow the controller to query switches about whether they support Scheduled Bundles, and to configure some of the switch parameters related to \emph{Scheduled Bundles}.

\ifdefined\ShortVersion
\else
\vspace{1mm}
\emph{3) Clock Synchronization: \rptp}
\vspace{1mm}

In the last decade PTP, based on the IEEE 1588~\cite{IEEE1588} standard, has become a common feature in commodity switches, typically providing a clock accuracy on the order of 1~microsecond. 

\ifdefined\BlindRev
\rptp~\cite{ispcsrptp} is a PTP variant designed for SDNs. 
\else
In~\cite{ispcsrptp,hotsdnrptp} we introduced \rptp\, a PTP variant for SDNs. 
\rptp\ is based on PTP, but is conceptually reversed.  
\fi
In PTP a single node periodically distributes its time to the other nodes in the network.
In \rptp\ all nodes in the network (the switches) periodically distribute their time to a single node (the controller). 
The controller keeps track of the offsets, denoted by $\oi$ for switch $i$, between its clock and each of the switches' clocks, and uses them to send each switch individualized timed commands. 

\rptp\ allows the complex clock algorithms to be implemented by the controller, whereas the `dumb' switches only need to distribute their time to the controller. Following the SDN paradigm, the \rptp\ algorithmic logic can be programmed and dynamically tuned at the controller without affecting the switches. 

Another advantage of \rptp, which played an important role in our experiments, is that \rptp\ allows the controller to keep track of the synchronization status of each clock; a clock synchronization protocol requires a long setup time, typically tens of minutes. \rptp\ provides an indication of when the setup process has completed.

\ifdefined\BlindRev
\else
\begin{figure}[htbp]
  \centering
  \fbox{\includegraphics[width=.25\textwidth]{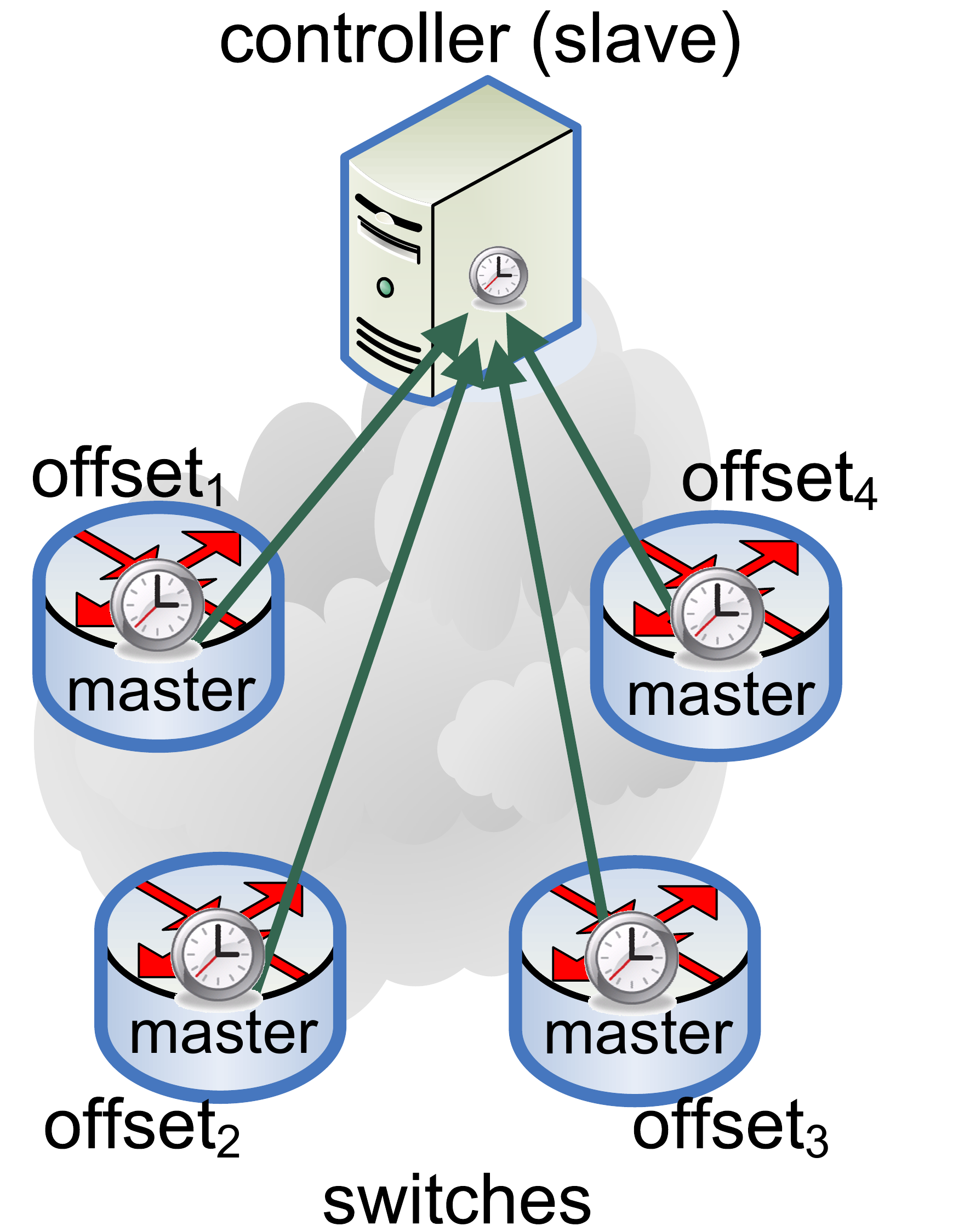}}
	\captionsetup{justification=raggedright}
  \caption{\rptp\ in SDN: switches distribute their time to the controller. Switches' clocks are \emph{not} synchronized. For every switch $i$, the controller knows $\oi$ between switch $i$'s clock and its local clock.}
  \label{fig:rptp}
\end{figure}

As shown in~\cite{ispcsrptp}, \rptp\ can be effectively used to perform timed updates; in order to have switch $i$ perform a command at time $T_s$, the controller instructs $i$ to perform the command at time $T^i_s$, where $T^i_s=T_s+\oi$ takes the offset between the controller and switch $i$ into account,\footnote{$T^i_s$, as described above is a first order approximation of the desired execution time. The controller can compute a more accurate execution time by also considering the clock skew and drift, as discussed in~\cite{ispcsrptp}.} 
causing $i$ to perform the action at time $T_s$ according to the controller's clock.
\fi 
\fi 

\ifdefined\ShortVersion
\subsection{Clock synchronization} 
\timec\ requires the switches and controller to maintain a local clock, enabling time-triggered events. Hence, the local clocks should be synchronized. The OpenFlow time extension we defined does not mandate a specific synchronization method. Various mechanisms may be used, e.g., the Network Time Protocol (NTP), the Precision Time Protocol (PTP)~\cite{IEEE1588}, or GPS-based synchronization. The prototype we designed and implemented uses \rptp~\cite{ispcsrptp}, a variant of PTP that is customized for SDN.
\fi

\subsection{Prototype Design and Implementation}
We have designed and implemented a software-based prototype of \timec, as illustrated in Fig.~\ref{fig:Arch}. The components we implemented are marked in black. 
\ifdefined\OpenSourceSoon
These components run on Linux, and will soon be publicly available as open source.  
\else
These components run on Linux, and are publicly available as open source~\cite{TimedSDNSource}.  
\fi

\begin{figure}[!t]
  \centering
  \fbox{\includegraphics[height=12\grafflecm]{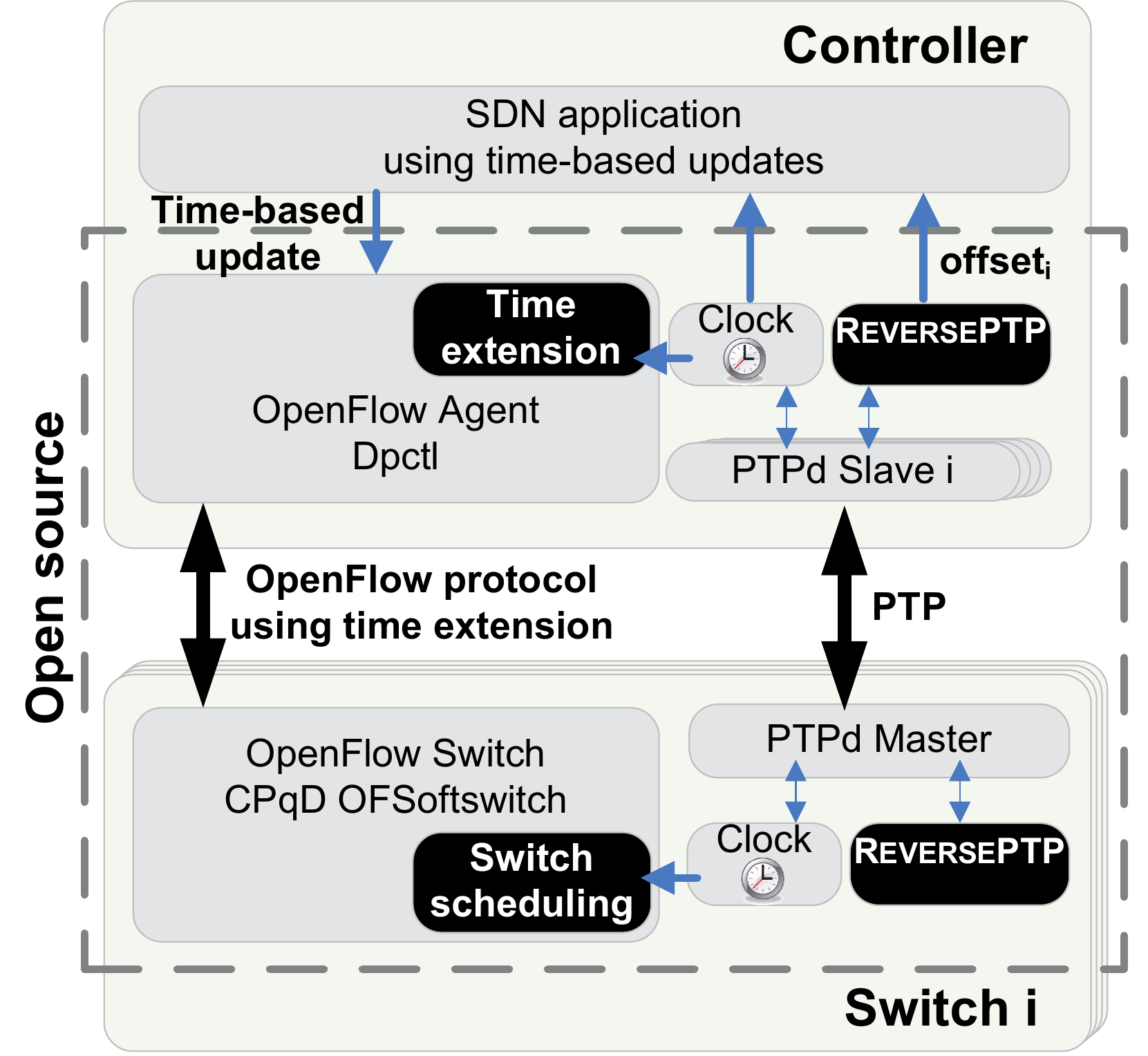}}
	\captionsetup{justification=raggedright}
  \caption{\timec\ prototype design: the black blocks are the components implemented in the context of this work.}
  \label{fig:Arch}
	\ifdefined\cutspace \vspace{-4mm} \fi
\end{figure}

Our \timec-enabled OFSoftswitch prototype was adopted by the ONF as the official prototype of Scheduled Bundles.\footnote{The ONF process for adding new features to OpenFlow requires every new feature to be prototyped.}

\textbf{Switches.} Every switch $i$ runs an OpenFlow switch software module. Our prototype is based on the open source CPqD OFSoftswitch~\cite{CPqDOF},\footnote{OFSoftswitch is one of the two software switches used by the Open Networking Foundation (ONF) for prototyping new OpenFlow features. We chose this switch since it was the first open source OpenFlow switch to include the Bundle feature.} incorporating the \emph{switch scheduling} module (see Fig.~\ref{fig:Arch}) that we implemented. When the switch receives a \emph{Scheduled Bundle} from the controller, the \emph{switch scheduling} module 
schedules the respective OpenFlow command to the desired time of execution. The switch scheduling module also handles \emph{Bundle Feature Request} messages received from the controller.

Each switch runs a \rptp\ master, which distributes the switch's time to the controller. Our \rptp\ prototype is a lightweight set of Bash scripts that is used as an abstraction layer over the well-known open source PTPd~\cite{ptpd} module. Our software-based implementation uses the Linux clock as the reference for PTPd, and for the switch's scheduling module. 
To the best of our knowledge, ours is the first open source implementation of \rptp.

\textbf{Controller.} The controller runs an OpenFlow agent, which communicates with the switches using the OpenFlow protocol. Our prototype uses the CPqD Dpctl (Datapath Controller), which is a simple command line tool for sending OpenFlow messages to switches. We have extended Dpctl by adding the time extension; the Dpctl command-line interface allows the user to define the execution time of a \emph{Bundle Commit}. Dpctl also allows a user to send a \emph{Bundle Feature Request} to switches.

The controller runs \rptp\ with $n$ instances of PTPd in slave mode, where $n$ is the number of switches in the network. One or more SDN applications can run on the controller and perform timed updates. The application can extract the offset, $\oi$, of every switch~$i$ from \rptp, and use it to compute the scheduled execution time of switch $i$ in every timed update. The Linux clock is used as a reference for PTPd, and for the SDN application(s).

\section{Evaluation}
\label{EvaluationSec}
\subsection{Evaluation Method}
\label{EvalMethodSec}
\textbf{Environment.}
We evaluated our prototype on a 71-node testbed in the DeterLab~\cite{DeterLabProj} environment. Each machine (PC) in the testbed either played the role of an OpenFlow switch, running our \timec-enabled prototype, or the role of a host, sending and receiving traffic. A separate machine was used as a controller, which was connected to the switches using an out-of-band network. 

\ifdefined\TechReport
We remark that we did not use Mininet~\cite{lantz2010network} in our evaluation, as Mininet is an emulation environment that runs on a single machine, making it impractical for emulating simultaneous or time-triggered events. We did, however, run our prototype over Mininet in some of our preliminary testing and verification.
\fi

\textbf{Performance attributes.}
Three performance attributes play a key role in our evaluation, as shown in Table~\ref{AttributeTable}.

\begin{table}[htbp]
		\centering
    \begin{tabular}{| l | p{7.2cm}|}
    \hline
    $\dmi$ & The average time elapsed between two consecutive messages sent by the controller. \\ 
	  $I_R$ & Installation latency range: the difference between the maximal rule installation latency and the minimal installation latency. \\ 
	  $\delta$ & Scheduling error: the maximal difference between the actual update time and the scheduled update time. \\ 
	  \hline
    \end{tabular}
    \caption{Performance Attributes.}
    \label{AttributeTable}
\end{table}

Intuitively, $\dmi$ and $I_R$ determine the performance of untimed updates. $\dmi$ indicates the \textbf{controller}'s performance; an OpenFlow controller can handle as many as tens of thousands~\cite{tavakoli2009applying} to millions~\cite{tootoonchian2012controller} of packets per second, depending on the type of controller and the machine's processing power. Hence, $\dmi$ can vary from 1 microsecond to several milliseconds. $I_R$ indicates the installation latency variation. The installation latency is the time elapsed from the instant the controller sends a rule modification message until the rule has been installed. The installation latency of an OpenFlow rule modification (FLOW\_MOD) has been shown to range from 1 millisecond to seconds~\cite{rotsos2012oflops,jin2014dynamic}, and grows dramatically with the number of installations per second. 

The attribute that affects the performance of \textbf{timed} updates is the switches' scheduling error, $\delta$. When an update is scheduled to be performed at time $T_0$, it is performed in practice at some time $t \in [T_0, T_0+\delta]$.\footnote{An alternative representation of the accuracy, $\delta$, assumes a symmetric error, $T_0 \pm \delta$. The two approaches are equivalent.}  The \schea, $\delta$, is affected by two factors: the device's \emph{clock accuracy}, which is the maximal offset between the clock value and the value of an accurate time reference, and the \emph{\exea}, which is a measure of how accurately the device can perform a timed update, given run-time parameters such as the concurrently executing tasks and the load on the device. The achievable clock accuracy strongly depends on the network size and topology, and on the clock synchronization method. For example, the clock accuracy using the Precision Time Protocol~\cite{IEEE1588} is typically on the order of 1~microsecond (e.g.,~\cite{ChinaMobile}).

\ifdefined\TechReport
\textbf{Software-based evaluation.}
Our experiments measure the three performance attributes in a setting that uses \textbf{software switches}.
While the \emph{values} we measured do not necessarily reflect on the performance of systems that use hardware-based switches, the merit of our evaluation is that we \textbf{vary} these parameters and analyze how they affect the network update performance with untimed approaches and with \timec.
\fi

\ifdefined\TechReport
\begin{figure*}[htbp]

	\centering
  \begin{subfigure}[t]{.33\textwidth}
  \centering
  \fbox{\includegraphics[height=6\grafflecm]{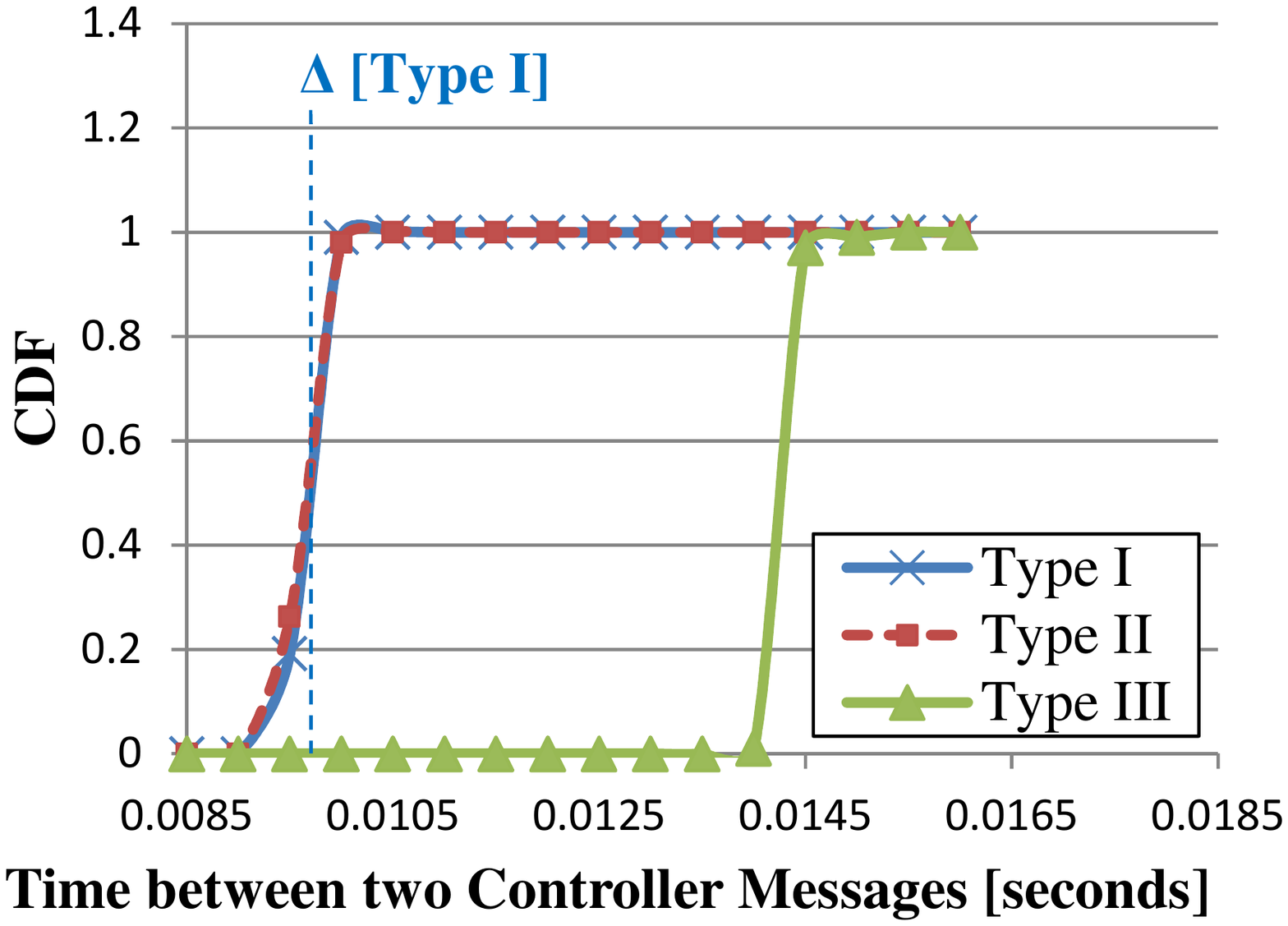}}
	\captionsetup{justification=raggedright}
  \caption{The empirical Cumulative Distribution Function (CDF) of the time elapsed between two consecutive controller messages. $\dmi$ is the average value, which is shown in the figure for Type I.}
  \label{fig:DeltaCDF}
  \end{subfigure}%
  \begin{subfigure}[t]{.33\textwidth}
  \centering
  \fbox{\includegraphics[height=6\grafflecm]{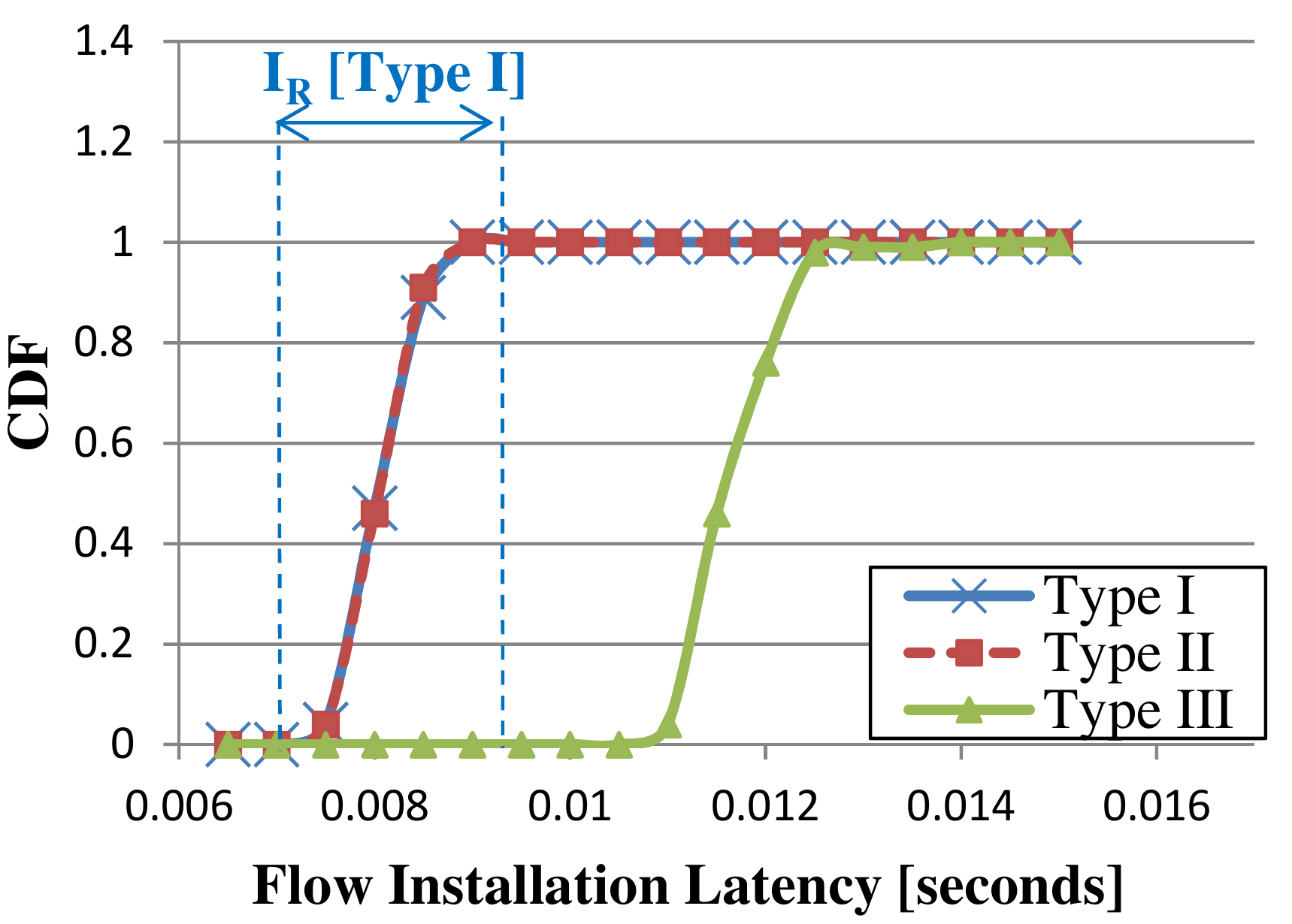}}
	\captionsetup{justification=centering}
  \caption{The empirical CDF of the flow installation latency. $I_R$ is the difference between the max and min values, as shown in the figure for Type I.}
  \label{fig:IvarCDF}
  \end{subfigure}%
  \begin{subfigure}[t]{.33\textwidth}
  \centering
  \fbox{\includegraphics[height=6\grafflecm]{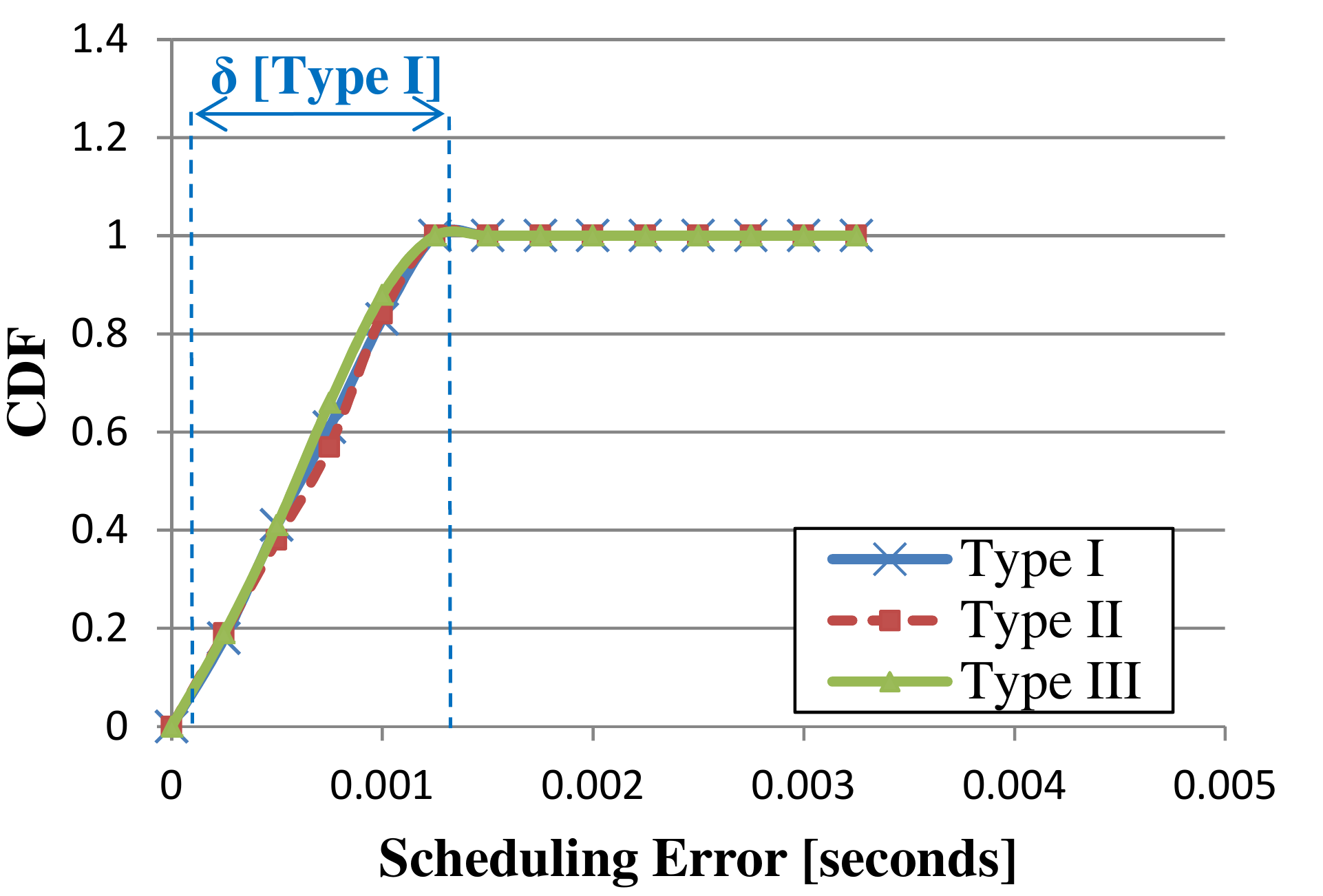}}
	\captionsetup{justification=raggedright}
  \caption{The empirical CDF of the scheduling error, i.e., the difference between the actual execution time and the scheduled execution time. $\delta$ is the maximal error value, as shown in the figure for Type I.}
  \label{fig:sm_deltaCDF}
  \end{subfigure}%

	\ifdefined\cutspace \vspace{-3mm} \fi

  \caption{Measurement of the three performance attributes: (a) $\dmi$, (b) $I_R$, and (c) $\delta$.}
  \label{fig:CDF}

	\ifdefined\cutspace \vspace{-3mm} \fi

\end{figure*}
\fi

\ifdefined\TechReport
\else
\begin{figure}[!b]
  \centering
  \fbox{\includegraphics[width=.35\textwidth]{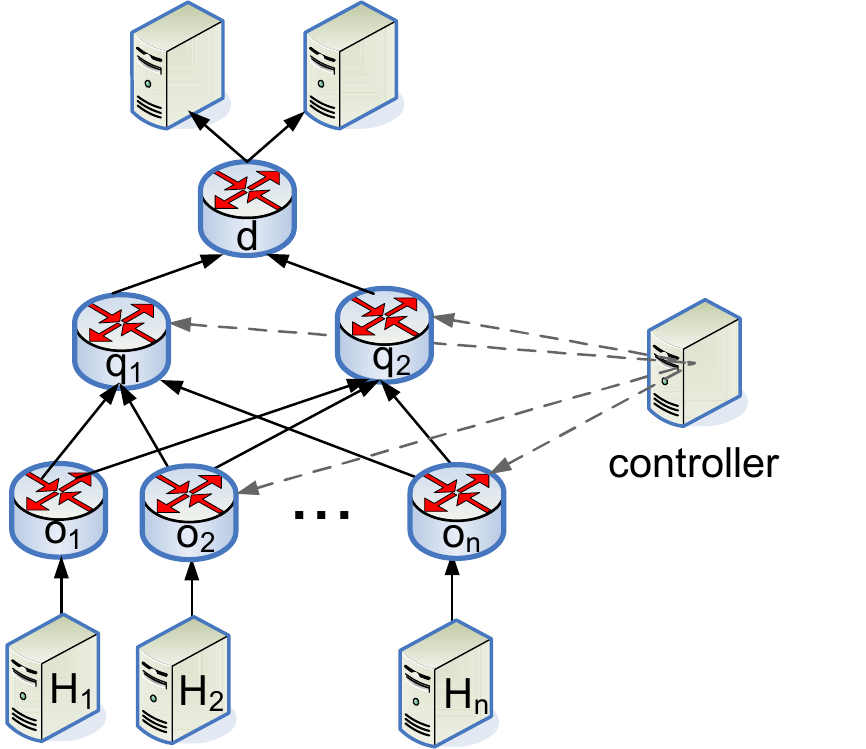}}
	\captionsetup{justification=centering}
  \caption{Experimental evaluation: every host and switch was emulated by a Linux machine in the DeterLab testbed. All links have a capacity of 10 Mbps. The controller is connected to the switches by an out-of-band network.}
  \label{fig:Experiment}
	\ifdefined\cutspace \vspace{-3mm} \fi
\end{figure}
\fi

\subsection{Performance Attribute Measurement} 
\begin{sloppypar}
Our experiments measured the three attributes, $\dmi$, $I_R$, and $\delta$, illustrating how accurately updates can be applied in software-based OpenFlow implementations. It should be noted that these three values depend on the processing power of the testbed machine; we measured the parameters for three types of DeterLab machines, Type I, II, and III, listed in Table~\ref{ParameterTable}. 
\ifdefined\TechReport
Each attribute was measured 100 times on each machine type, and Fig.~\ref{fig:CDF} illustrates our results. The figure graphically depicts the values $\dmi$, $I_R$, and $\delta$ of machine Type I as an example.
\fi
\end{sloppypar}

\begin{table}[!t]
		\centering
    \begin{tabular}{|l|l||c|c|c|}
    \hline
    \multicolumn{2}{|l||}{Machine Type} & $\dmi$ & $I_R$ & $\delta$ \\ \hline \hline
    \multirow{2}{*}{I} & Intel Xeon E3 LP & \multirow{2}{*}{9.64} & \multirow{2}{*}{1.3} & \multirow{2}{*}{1.23} \\ 
     & 2.4 GHz, 16 GB RAM &  & &   \\ \hline 
    \multirow{2}{*}{II} & Intel Xeon &  \multirow{2}{*}{9.6} & \multirow{2}{*}{1.47} & \multirow{2}{*}{1.18}  \\ 
     & 2.1 GHz, 4 GB RAM &  & &   \\ \hline 
    \multirow{2}{*}{II} & Intel Dual Xeon &  \multirow{2}{*}{14.27} & \multirow{2}{*}{2.72} & \multirow{2}{*}{1.19}  \\ 
     & 3 GHz, 2 GB RAM &  & &   \\ \hline 
    \end{tabular}
    \caption{Measured attributes in milliseconds.}
    \label{ParameterTable}
		\ifdefined\cutspace \vspace{-4mm} \fi
\end{table}

\ifdefined\TechReport
The measured scheduling error, $\delta$, was slightly more than 1~millisecond in all the machines we tested. Our experiments showed that the \emph{clock accuracy} using \rptp\ over the DeterLab testbed is on the order of $100$ microseconds. The measured value of~$\delta$ in Table~\ref{ParameterTable} shows the \emph{execution accuracy}, which is an order of magnitude higher. The installation latency range, $I_R$, was slightly higher than $\delta$, around $1$~to~$3$ milliseconds. The measured value of $\dmi$ was high, on the order of 10~milliseconds, as Dpctl is not optimized for performance.
\fi

In software-based switches, the CPU handles both the data-plane traffic and the communication with the controller, and thus $I_R$ and $\delta$ can be affected by the rate of data-plane traffic through the switch. Hence, in our experiments we fixed the rate of traffic through each switch to 10 Mbps, allowing an `apples-to-apples' comparison between experiments. 

\ifdefined\ShortVersion
\else
\subsection{Microbenchmark: Video Swapping}
\label{MicrobSec}
To demonstrate how \timec\ is used in a real-life scenario, we reconstructed the video swapping topology of~\cite{edwards2014using}, as illustrated in Fig.~\ref{fig:FoxTopo}. Two video cameras, A and B, transmit an uncompressed video stream to targets A and B, respectively. At a given point in time, the two video streams are swapped, so that the stream from source A is transmitted to target B, and the stream from B is sent to target A. As described in~\cite{edwards2014using}, the swap must be performed at a specific time instant, in which the video sources transmit data that is not visible to the viewer, making the swap unnoticeable. 

\begin{figure}[htbp]

	\centering
  \begin{subfigure}[t]{.26\textwidth}
  \centering
  \fbox{\includegraphics[height=5.7\grafflecm]{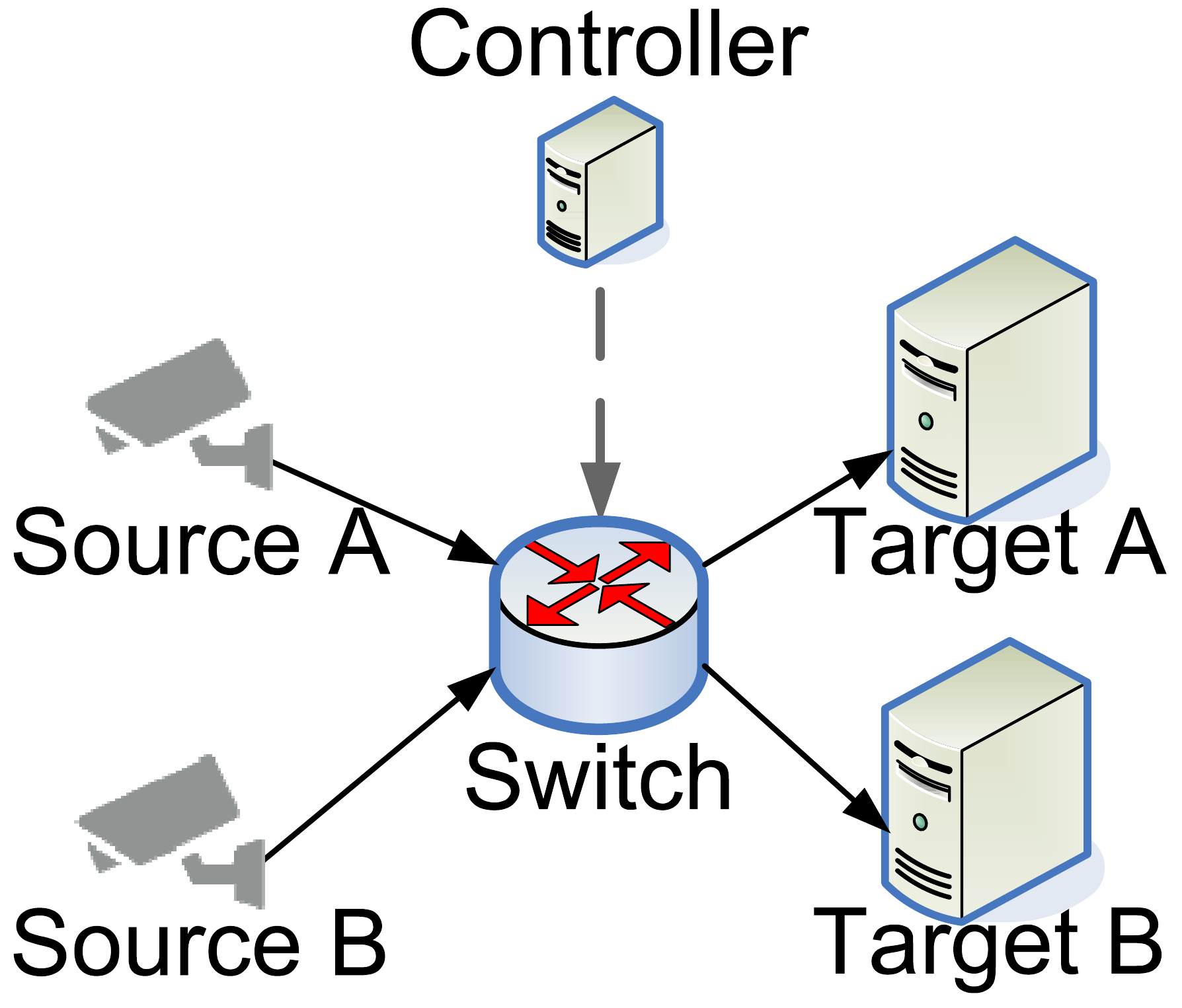}}
	\captionsetup{justification=centering}
  \caption{Topology.}
  \label{fig:FoxTopo}
  \end{subfigure}%
  \begin{subfigure}[t]{.22\textwidth}
	\centering
  \fbox{\includegraphics[height=5.7\grafflecm]{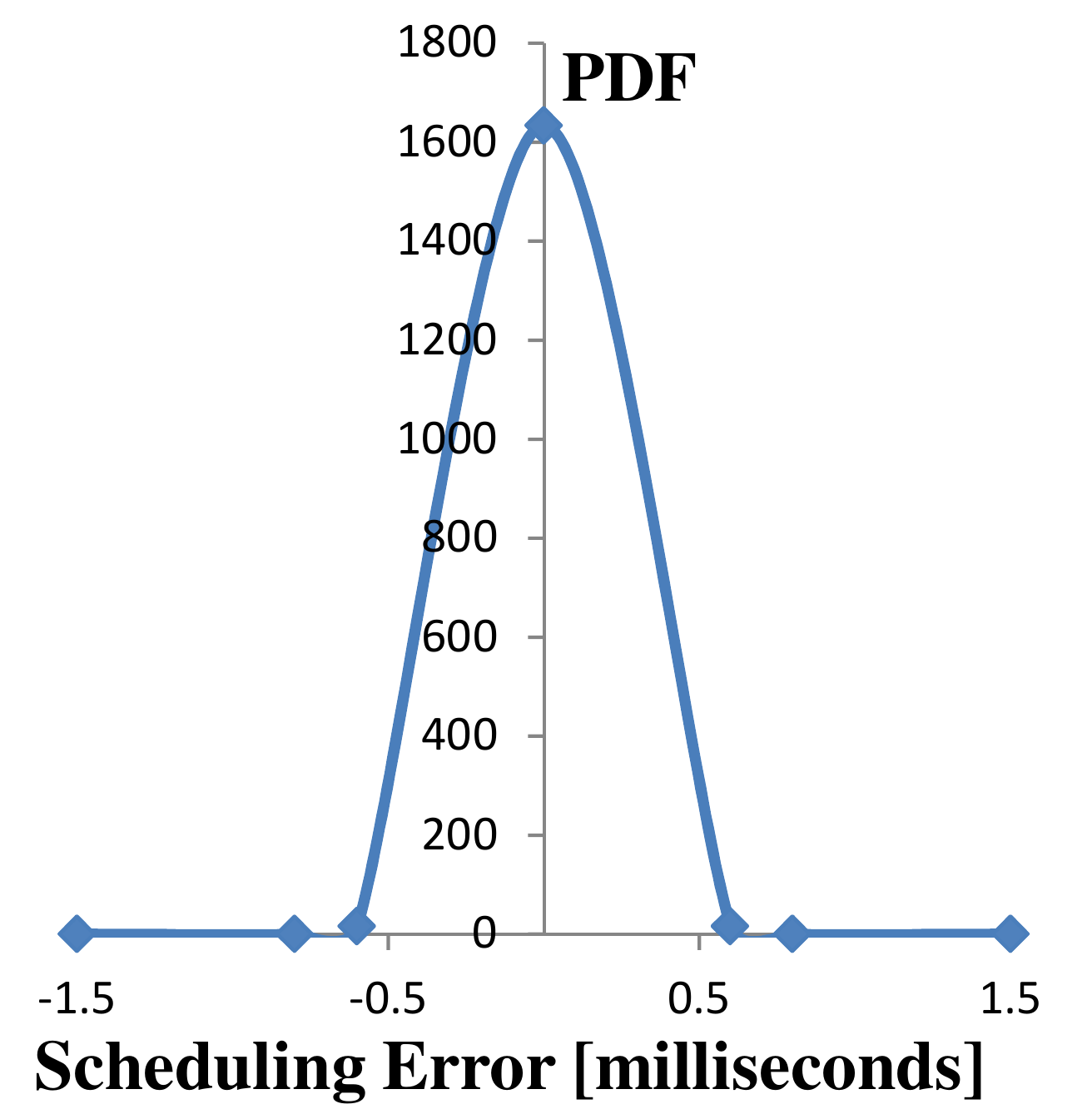}}
	\captionsetup{justification=centering}
  \caption{Video swapping accuracy.}
  \label{fig:FoxPDF}
  \end{subfigure}%

  \caption{Microbenchmark: video swapping.}
  \label{fig:Fox}

\end{figure}

The authors of~\cite{edwards2014using} noted that the precisely-timed swap cannot be performed by an OpenFlow switch, as currently OpenFlow does not provide abstractions for performing accurately timed changes. Instead, it uses \emph{source timing}, where sources A and B are time-synchronized, and determine the swap time by using a swap indication in the packet header. The OpenFlow switch acts upon the swap indication to determine the correct path for each stream. We note that the main drawback of this source-timed approach is that the SMPTE 2022-6 video streaming standard~\cite{SMPTE}, which was used in~\cite{edwards2014using}, does not currently define an indication about where in the video stream a packet comes from, and specifically does not include an indication about the correct swapping time. Hence, off-the-shelf streaming equipment does not provide this indication. In~\cite{edwards2014using}, the authors used a dedicated Linux server to integrate the non-standard swap indication.

In this experiment we studied how \timec\ can tackle the video swapping scenario, avoiding the above drawback. Each node in the topology of Fig.~\ref{fig:FoxTopo} was emulated by a DeterLab machine. We used two 10 Mbps flows, generated by Iperf~\cite{Iperf}, to simulate the video streams. Each swap was initiated by the controller 100 milliseconds in advance (as in~\cite{edwards2014using}): the controller sent a Scheduled Bundle, incorporating two updates, one for each of the flows. We repeated the experiment 100 times, and measured the scheduling error. 

The measurement was performed by analyzing capture files taken at the sources and at the switch's egress ports. A swap that was scheduled to be performed at time $T$, was considered accurate if every packet that was transmitted by each of the source before time $T$ was forwarded according to the old configuration, and every packet that was transmitted after $T$ was forwarded according to the new configuration. The scheduling error of each swap (measured in milliseconds) was computed as the number of misrouted packets, divided by the bandwidth of the traffic flow. The sign of the scheduling error indicates whether the swap was performed before the scheduled time (negative error) or after it (positive error).

Fig.~\ref{fig:FoxPDF} illustrates the empirical Probability Density Function (PDF) of the scheduling error of the swap, i.e., the difference between the actual swapping time and the scheduled swapping time. As shown in the figure, the swap is performed within $\pm 0.6$ milliseconds of the scheduled swap time. We note that this is the achievable accuracy in a software-based OpenFlow switch, 
and that a much higher degree of accuracy, on the order of microseconds, can be achieved if two conditions are met: (i) A hardware switch is used, supporting timed updates with a microsecond accuracy, as shown in~\cite{Infocom-TimeFlip}, and (ii) The cameras are connected to the switch over a single hop, allowing low latency variation, on the order of microseconds.

\fi 

\begin{figure*}[!t]

	\centering
  \begin{subfigure}[t]{.22\textwidth}
  \centering
  \fbox{\includegraphics[height=5\grafflecm]{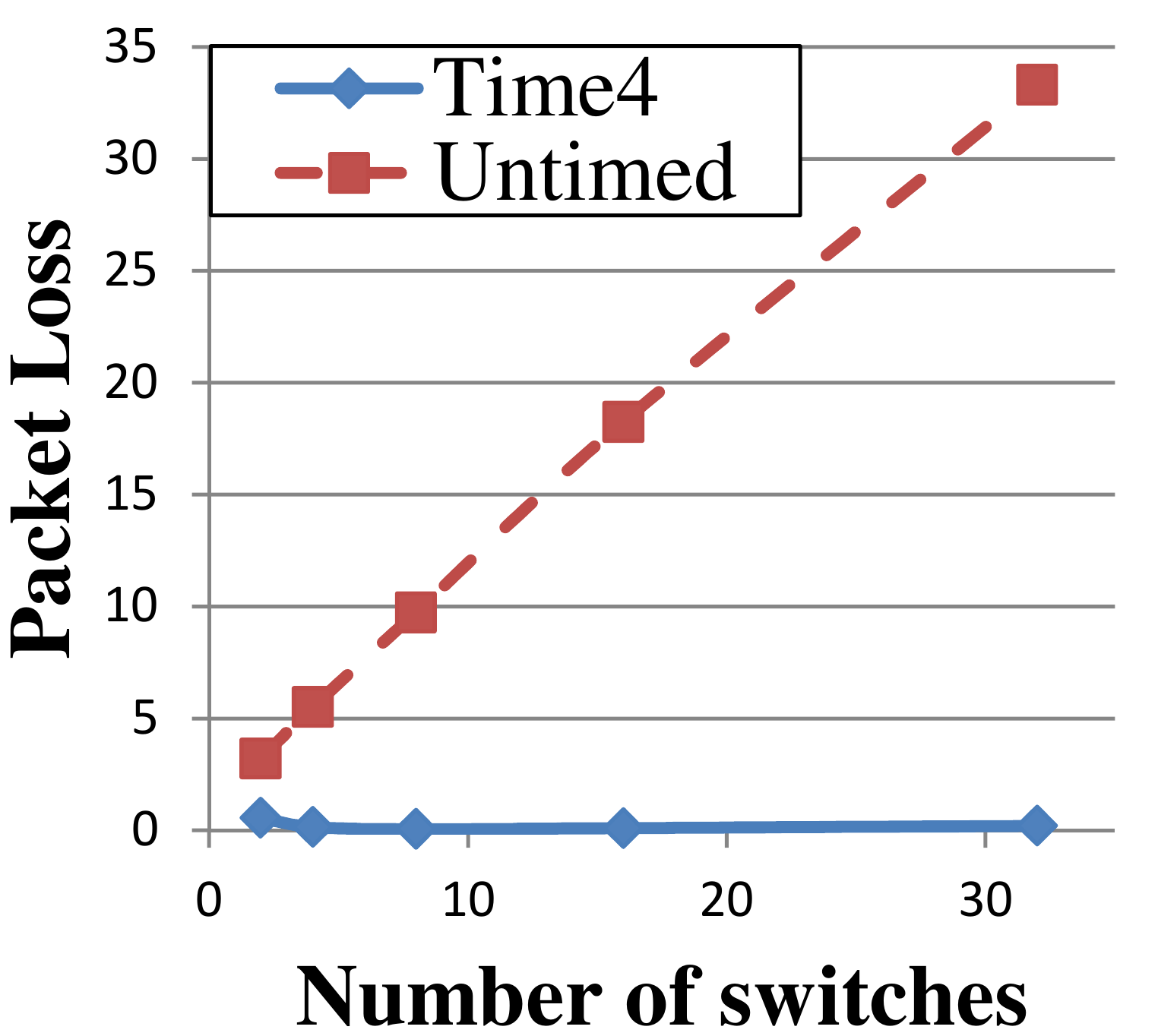}}
  \ifdefined\cutspace \vspace{-2mm} \fi
	\captionsetup{justification=centering}
  \caption{The no. of packets lost in~a flow swap vs. no. of switches involved in the update.}
  \label{fig:LossVsN}
  \end{subfigure}%
  \begin{subfigure}[t]{.3025\textwidth}
  \centering
  \fbox{\includegraphics[height=5\grafflecm]{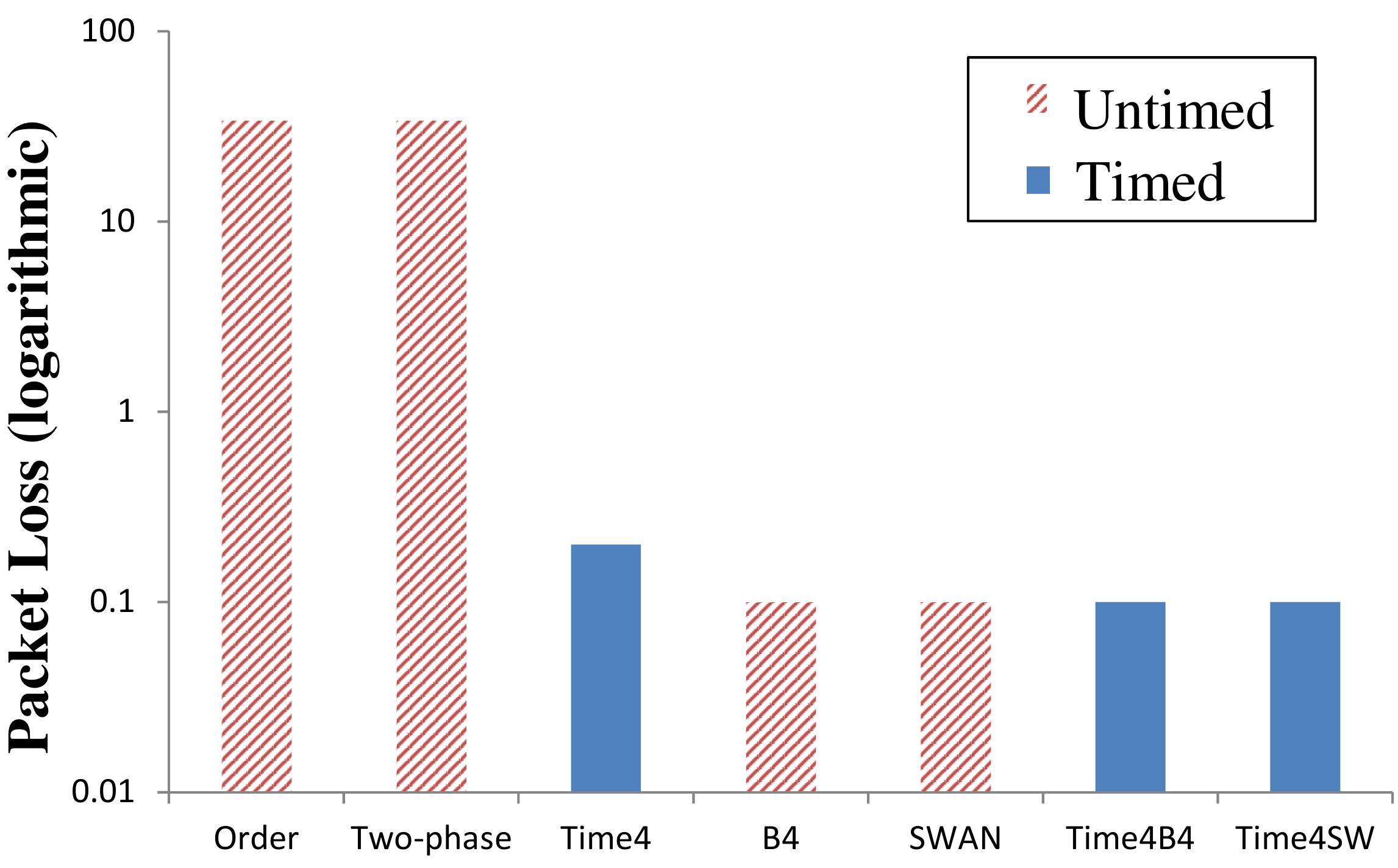}}
  \ifdefined\cutspace \vspace{-5mm} \fi
	\captionsetup{justification=centering}
  \caption{The number of packets lost in~a flow swap in different update approaches (with~$n=32$).}
  \label{fig:UpdateCompare}
  \end{subfigure}%
  \begin{subfigure}[t]{.24\textwidth}
  \centering
  \fbox{\includegraphics[height=5\grafflecm]{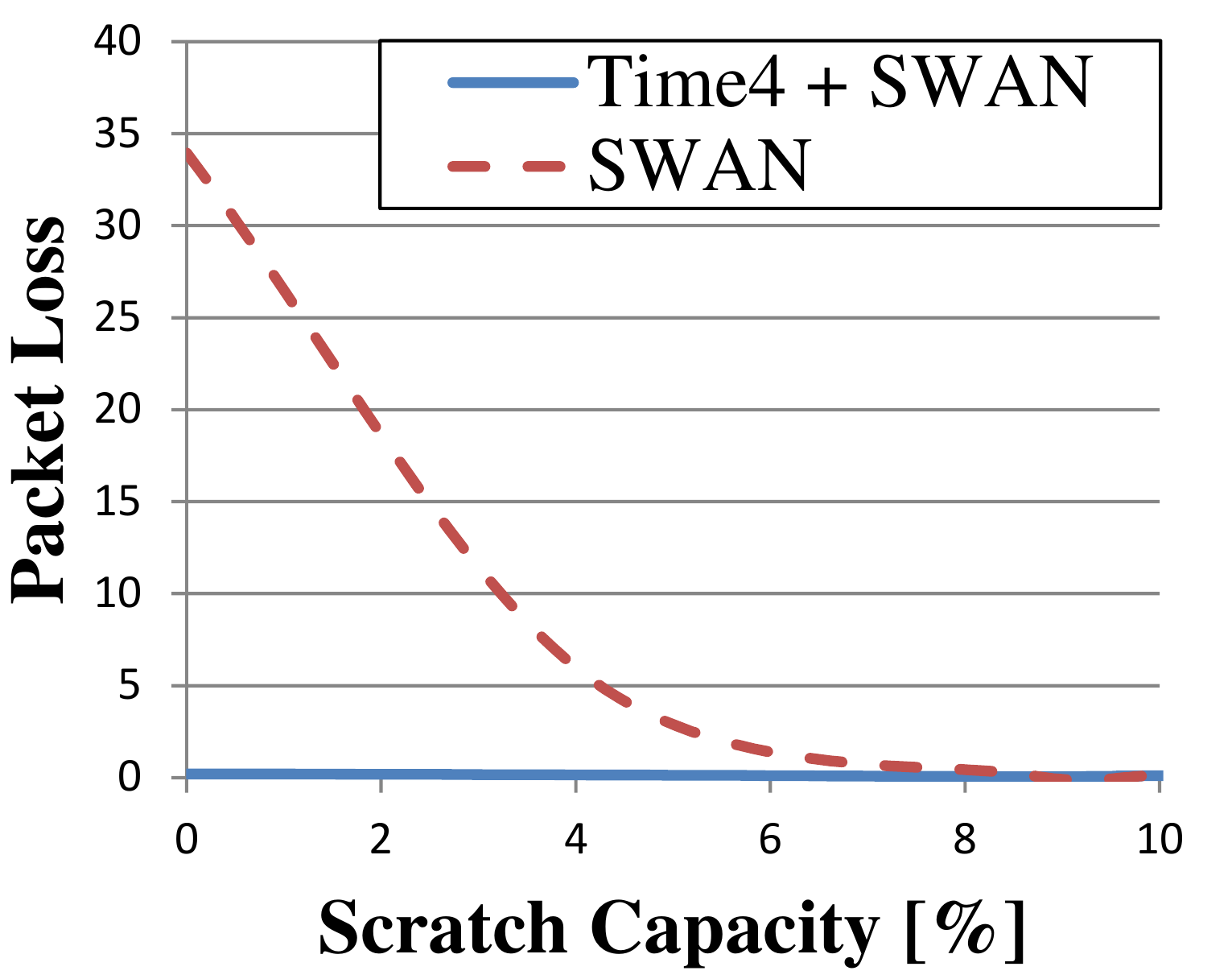}}
  \ifdefined\cutspace \vspace{-5mm} \fi
	\captionsetup{justification=centering}
  \caption{The number of packets lost in~a flow swap using SWAN and \timec\ $+$ SWAN (with~$n=32$).}
  \label{fig:SWAN}
  \end{subfigure}%
  \begin{subfigure}[t]{.24\textwidth}
  \centering
  \fbox{\includegraphics[height=5\grafflecm]{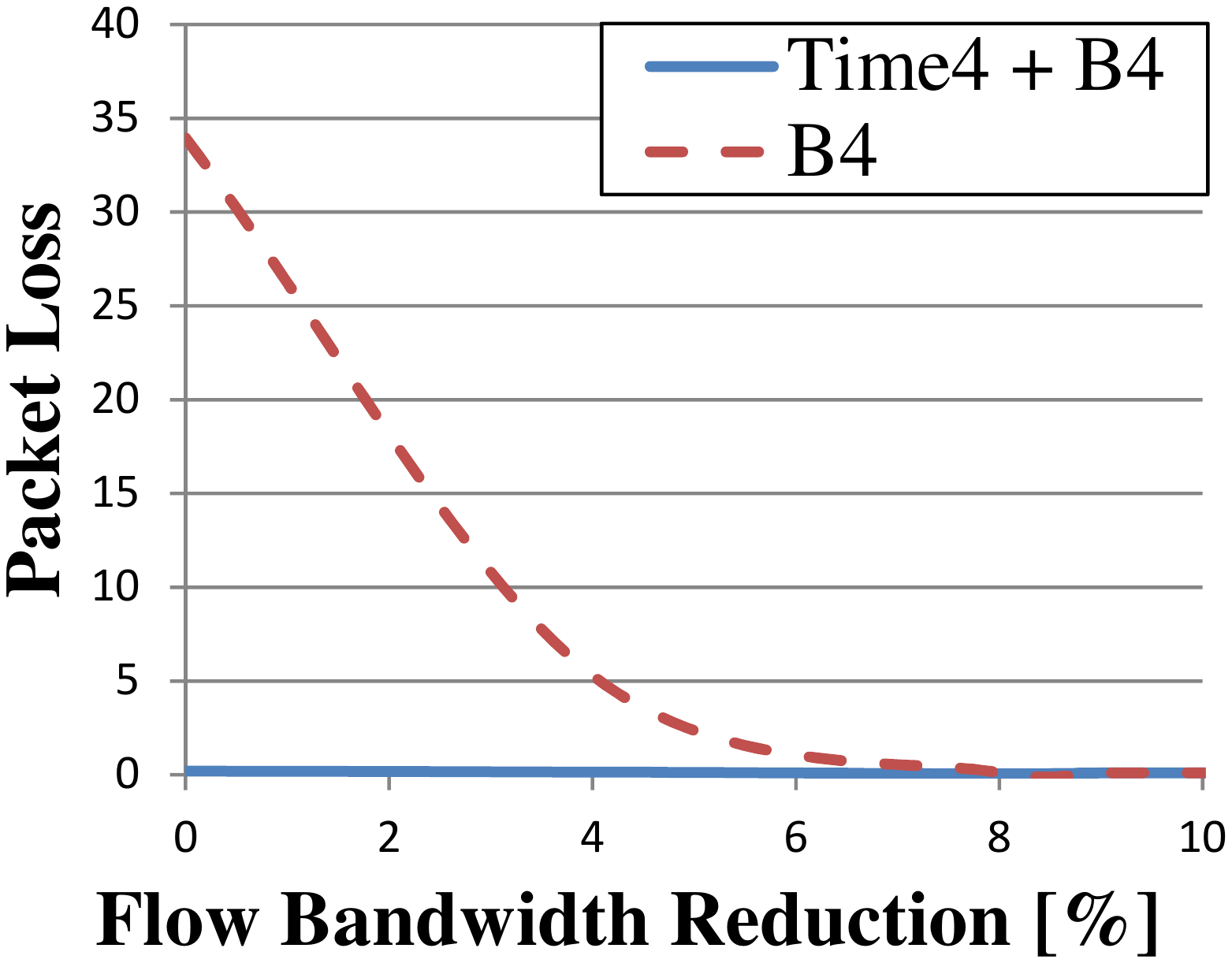}}
  \ifdefined\cutspace \vspace{-5mm} \fi
	\captionsetup{justification=centering}
  \caption{The number of packets lost in~a flow swap using B4 and \timec\ $+$ B4 (with~$n=32$).}
  \label{fig:B4}
  \end{subfigure}%

	\ifdefined\cutspace \vspace{-1mm} \fi
  \caption{Flow swap performance: in large networks (a) \timec\ allows significantly less packet loss than untimed approaches. The packet loss of \timec\ is slightly higher than SWAN and B4 (b), while the latter two methods incur higher overhead. Combining \timec\ with SWAN or B4 provides the best of both worlds; low packet loss (b) and low overhead (c~and~d).}
  \ifdefined\cutspace \vspace{-4mm} \fi
	\label{fig:Nswap}

\end{figure*}

\subsection{Flow Swap Evaluation}
\label{FSAnalysisSec}

\emph{1) Experiment Setting}
\vspace{0.5mm}

We evaluated our prototype on a 71-node testbed under. We used the testbed to emulate an OpenFlow network with 32 hosts and 32 leaf switches, as depicted in Fig.~\ref{fig:Experiment}, with $n=32$.

\ifdefined\TechReport
\begin{figure}[htbp]
  \centering
  \fbox{\includegraphics[width=.35\textwidth]{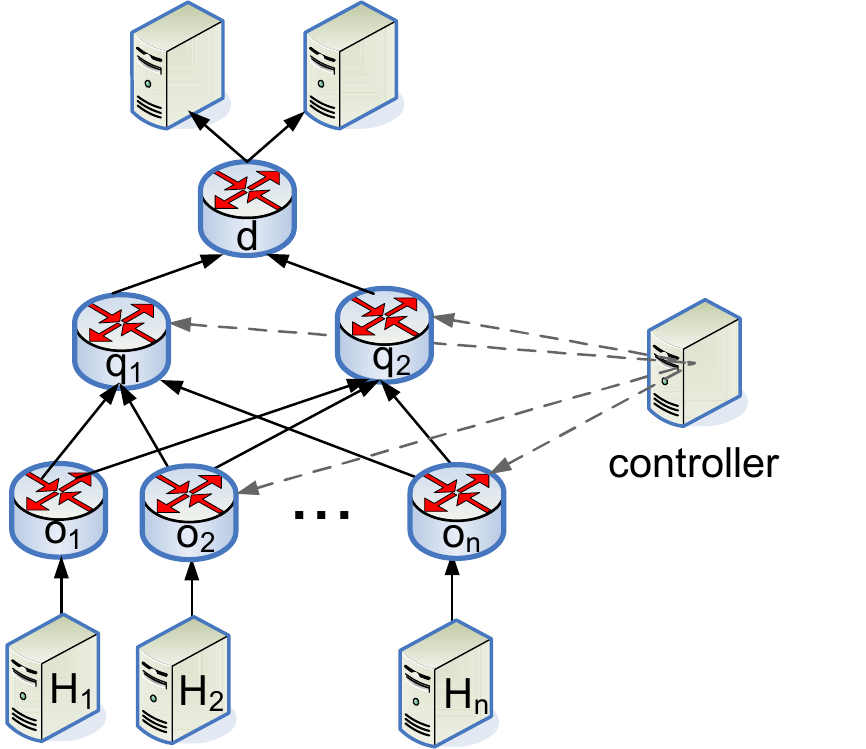}}
	\captionsetup{justification=centering}
  \caption{Experimental evaluation: every host and switch was emulated by a Linux machine in the DeterLab testbed. All links have a capacity of 10 Mbps. The controller is connected to the switches by an out-of-band network.}
  \label{fig:Experiment}
	\ifdefined\cutspace \vspace{-3mm} \fi
\end{figure}
\fi


\textbf{Metric.}
A flow swap that is not performed in a coordinated way may bare a high cost: either packet loss, deep buffering, or a combination of the two. We use packet loss as a \textbf{metric} for the cost of flow swaps, assuming that deep buffering is not used.

We used Iperf to generate flows from the sources to the destination, and to measure the number of packets lost between the source and the destination.

\begin{sloppypar}
\textbf{The flow swap scenario.} All experiments were flow swaps with a swap impact of $0.5$.\footnote{By Theorem~\ref{SwapImpactTheorem}, the source can force the controller to perform a flow swap with an impact as high as roughly $0.5$.} We used two static flows, which were not reconfigured in the experiment: $H_1$ generates a~$5$~Mbps flow that is forwarded through~$q_1$, and~$H_2$ generates a $5$~Mbps flow that is generated through~$q_2$. We generated $n$ additional flows (where $n$ is the number of switches at the bottom layer of the graph): (i) A~$5$~Mbps flow from $H_1$ to the destination. (ii) $n-1$ flows, each having a bandwidth of $\frac{5}{n-1}$~Mbps. Every flow swap in our experiment required the flow of (i) to be swapped with the $n-1$ flows of (ii). Note that this swap has an impact of $0.5$.
\end{sloppypar}

\vspace{2mm}
\emph{2) Experimental Results}
\vspace{0.5mm}

\textbf{\boldtimec\ vs. other update approaches.} 
In this experiment we compared the packet loss of \timec\ to other update approaches described in Sec.~\ref{RelatedSec}. As discussed in Sec.~\ref{RelatedSec}, applying the \order\ approach or the \twophase\ approach to flow swaps produces similar results. This observation is illustrated in Fig.~\ref{fig:UpdateCompare}. In the rest of this section we refer to these two approaches collectively as the \emph{untimed} approaches. 

In our experiments we also implemented a SWAN-based~\cite{hong2013achieving} update, and a B4-based~\cite{jain2013b4} update. 
In SWAN, we used a 10\% scratch on each of the links, and in B4 updates we temporarily reduced the bandwidth of each flow by 10\% to avoid packet loss. As depicted in Fig.~\ref{fig:UpdateCompare}, SWAN and B4 yield a slightly lower packet loss rate than \timec; the average number of packets lost in each \timec\ flow swap is $0.2$, while with SWAN and B4 only $0.1$ packets are lost on average. 

To study the effect of using \textbf{time} in SWAN and in B4, we also performed hybrid updates, illustrated in Fig.~\ref{fig:SWAN} and~\ref{fig:B4}, and in the two right-most bars of Fig.~\ref{fig:UpdateCompare}. We combined SWAN and \timec, by performing a timed update on a network with scratch capacity, and compared the packet loss to the conventional SWAN-based update. We repeated the experiment for various values of scratch capacity, from 0\% to 10\%. As illustrated in Fig.~\ref{fig:SWAN}, the \timec+SWAN approach can achieve the same level of packet loss as SWAN with \textbf{less scratch capacity}. We performed a similar experiment with a timed B4 update, varying the bandwidth reduction rate between 0\% and 10\%, and observed similar results.

\textbf{Number of switches.}
We evaluated the effect of~$n$, the number of switches involved in the flow swap, on the packet loss. 
We performed an $n$-swap with $n=2,4,8,16,32$. As illustrated in Fig.~\ref{fig:LossVsN}, the number of packets lost during an untimed update grows linearly with the number of switches $n$, while the number of packets lost in a \timec\ update is less than one on average, and is not affected by the number of switches. As $n$ increases, the update duration\footnote{The \textbf{update duration} is the time elapsed from the instant the first switch is updated until the instant the last switch is updated. In our setting the update duration is roughly $(n-1)\Delta$.} is longer, and hence more packets are lost during the update procedure.

\ifdefined\ShortVersion
\else
\textbf{Controller performance.}
In this experiment we explored how the controller's performance, represented by $\dmi$, affects the packet loss rate in an untimed update. As~$\dmi$ increases, the update procedure requires a longer period of time, and hence more packets are lost (Fig.~\ref{fig:LossvsDelta}) during the process. We note that although previous work has shown that $\Delta$ can be on the order of microseconds in some cases~\cite{tootoonchian2012controller}, Dpctl is not optimized for performance, and hence $\Delta$ in our experiments was on the order of milliseconds. As shown in Fig.~\ref{fig:LossvsDelta}, we synthetically increased $\Delta$, and observed its effect on the packet loss during flow swaps.

\begin{figure}[htbp]

	\centering
  \fbox{\includegraphics[height=4.9\grafflecm]{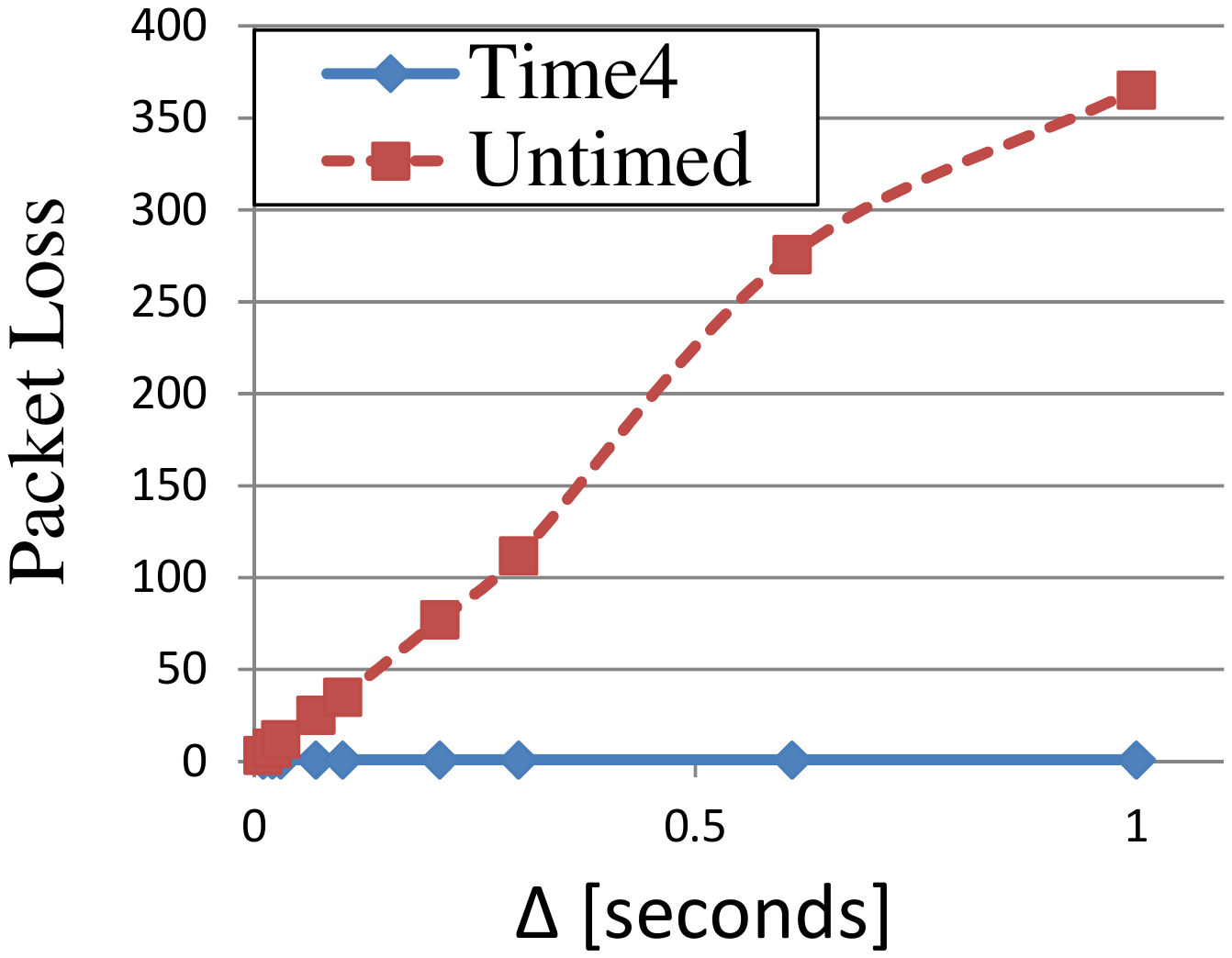}}
	\captionsetup{justification=centering}
	\ifdefined\TechReport
  \caption{The number of packets lost in~a flow swap vs. $\Delta$. The packet loss in \timec\ is not affected by the controller's performance ($\Delta$).}
	\else
  \caption{The number of packets lost in~a flow swap vs. $\Delta$.}
	\fi
  \label{fig:LossvsDelta}
\end{figure}

\fi

\textbf{Installation latency variation.}
Our next experiment (Fig.~\ref{fig:LossvsIr}) examined how the installation latency variation, denoted by $\ivar$, affects the packet loss during an untimed update. We analyzed different values of $\ivar$: in each update we synthetically determined a uniformly distributed installation latency, $I\sim U[0,\ivar]$. 
As shown in Fig.~\ref{fig:LossvsIr}, the switch's installation latency range, $\ivar$, dramatically affects the packet loss rate during an untimed update. Notably, when $\ivar$ is on the order of 1~second, as in the extreme scenarios of~\cite{rotsos2012oflops,jin2014dynamic}, \timec\ has a significant advantage over the untimed approach.

\begin{figure}[htbp]

	\centering
  \begin{subfigure}[t]{.235\textwidth}
  \centering
  \fbox{\includegraphics[height=4.6\grafflecm]{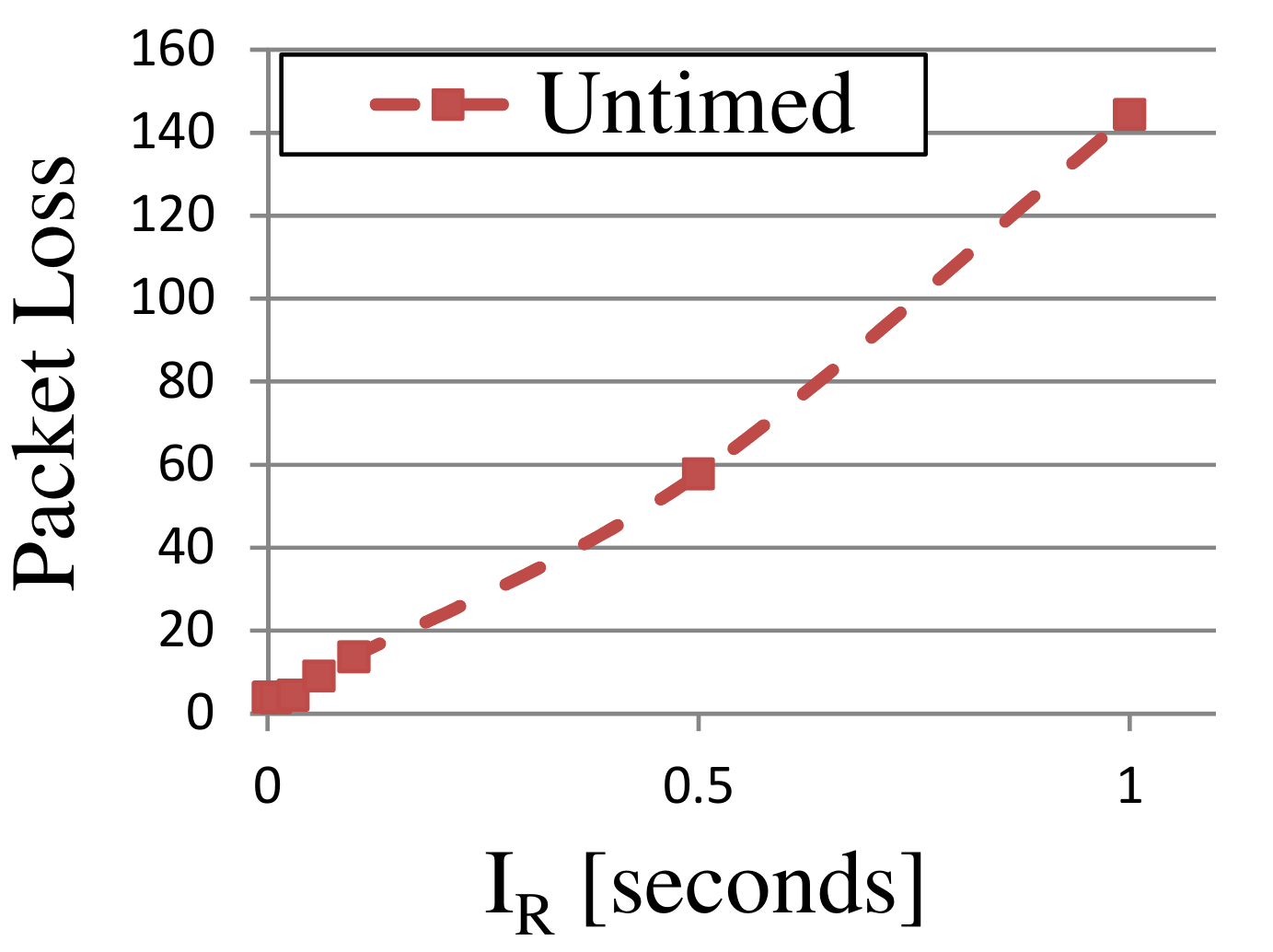}}
	\captionsetup{justification=centering}
  \caption{The number of packets lost in~a flow swap vs. the installation~latency~range, $\ivar$.}
  \label{fig:LossvsIr}
  \end{subfigure}%
  \begin{subfigure}[t]{.25\textwidth}
  \centering
  \fbox{\includegraphics[height=4.6\grafflecm]{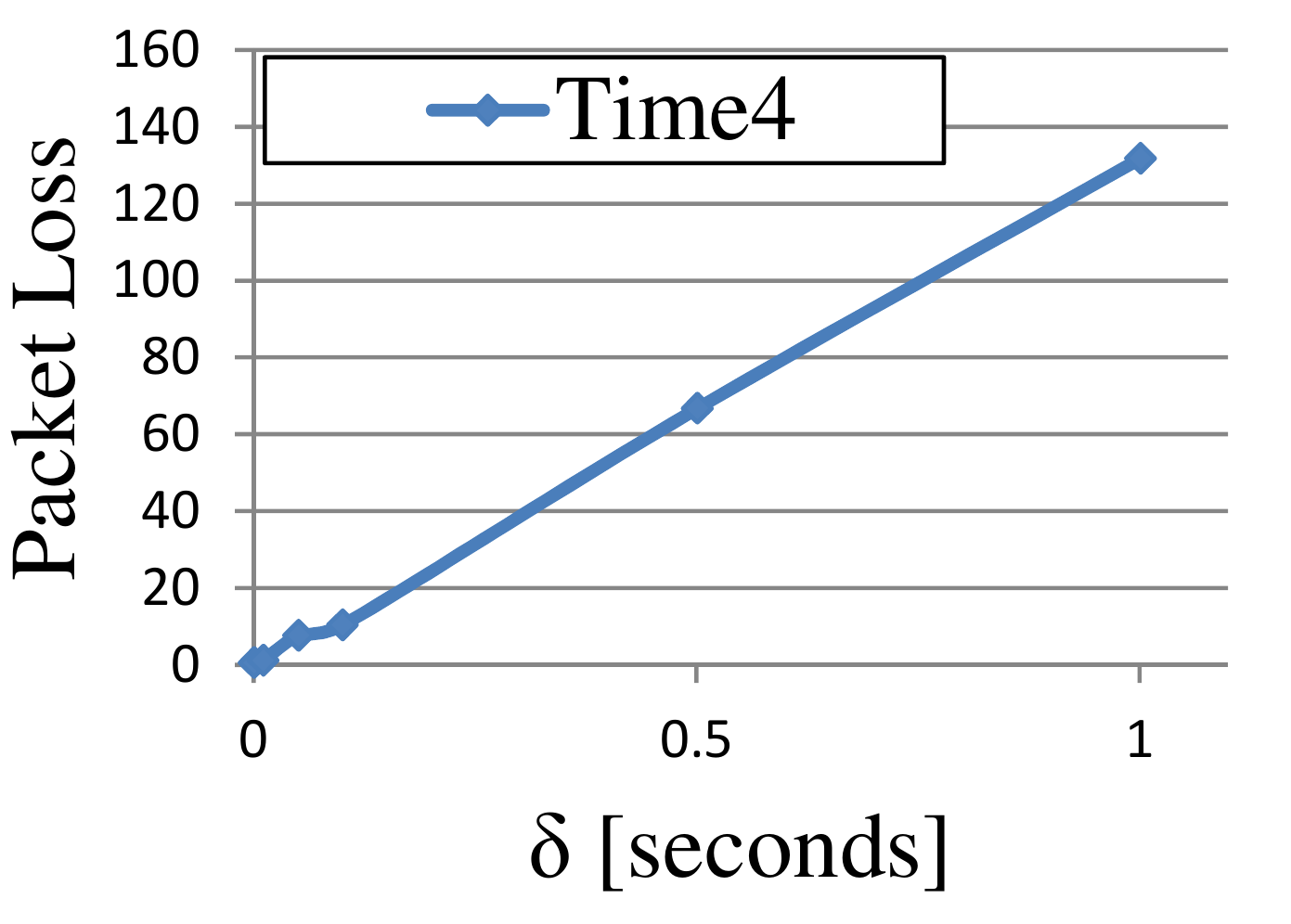}}
	\captionsetup{justification=centering}
  \caption{The number of packets lost in~a flow swap vs. the scheduling~error, $\delta$.}
  \label{fig:Lossvs_delta}
  \end{subfigure}%

	\ifdefined\TechReport
  \caption{Performance as a function of $\ivar$ and $\delta$. Untimed updates are affected by the installation latency variation ($I_R$), whereas \timec\ is affected by the scheduling error ($\delta$). \timec\ is advantageous since typically $\delta < I_R$. }
	\else
  \caption{Performance as a function of $\ivar$ and $\delta$.}
	\fi
  \label{fig:Ivar}
	\ifdefined\cutspace \vspace{-3mm} \fi

\end{figure}

\textbf{Scheduling error.}
Figure~\ref{fig:Lossvs_delta} depicts the packet loss as a function of the scheduling error of \timec.
By Fig.~\ref{fig:LossVsN}, ~\ref{fig:LossvsIr} and~\ref{fig:Lossvs_delta}, we observe that if $\delta$ is sufficiently low compared to $\ivar$ and $(n-1) \Delta$, then \timec\ outperforms the untimed approaches.  Note that even if switches are not implemented with extremely low scheduling error $\delta$, we expect \timec\ to outperform the untimed approach, as typically $\delta < \ivar$, as further discussed in Section~\ref{DiscussionSec}.

\textbf{Summary.}
The experiments presented in this section demonstrate that \timec\ performs significantly better than untimed approaches, especially when the update involves multiple switches, or when there is a non-deterministic installation latency. Interestingly, \timec\ can be used in conjunction with existing approaches, such as SWAN and B4, allowing the same level of packet loss \textbf{with less overhead} than the untimed variants.

\ifdefined\cutspace \vspace{-2mm} \fi
\section{Discussion}
\label{DiscussionSec}

\ifdefined\ShortVersion
\textbf{Scheduling accuracy.} 
\else
\emph{1) Scheduling accuracy}
\vspace{1mm}

\fi
The advantage of timed updates greatly depends on the \textbf{scheduling accuracy}, i.e., on the switches' ability to accurately perform an update at its scheduled time. Clocks can typically be synchronized on the order of $1$~microsecond (e.g.,~\cite{ChinaMobile}) using PTP~\cite{IEEE1588}. However, a switch's ability to accurately perform a scheduled action depends on its implementation. 

\begin{itemize}
	\item \emph{Software switches:} Our experimental evaluation showed that the scheduling error in the software switches we tested was on the order of 1~millisecond.
	\item \emph{Hardware-based scheduling:} The work of~\cite{Infocom-TimeFlip} has shown a method that allows the scheduling error of timed events in hardware switches to be as low as 1~microsecond.
\ifdefined\TechReport
	\item \emph{Software-based scheduling in hardware switches:} 
A scheduling mechanism that relies on the switch's software may be affected by the switch's operating system and by other running tasks. 
Measures can be taken to implement an accurate software-based scheduling in \timec: when a switch is aware of an update that is scheduled to take place at time $T_s$, it can avoid performing heavy maintenance tasks at this time, such as TCAM entry rearrangement. Update messages received slightly before time $T_s$ can be queued and processed after the scheduled update is executed. Moreover, if a switch receives a timed command that is scheduled to take place at the same time as a previously received command, it can send an error message to the controller, indicating that the last received command cannot be executed. 
\fi
\end{itemize}

It is an important observation that in a typical system we expect the scheduling error to be lower than the installation latency variation, i.e., $\delta < I_R$. Untimed updates have a non-deterministic installation latency. On the other hand, timed updates are predictable, and can be scheduled in a way that avoids conflicts between multiple updates, allowing $\delta$ to be typically lower than $I_R$. 

\ifdefined\ShortVersion
\textbf{Model assumptions.} 
\else
\vspace{2mm}
\emph{2) Model assumptions}
\vspace{1mm}

\fi
Our model assumes a \emph{lossless} network with \emph{unsplittable}, \emph{fixed-bandwidth} flows. 
A notable example of a setting in which these assumptions are often valid is a WAN or a carrier network. In carrier networks the maximal \textbf{bandwidth} of a service is defined by its bandwidth profile~\cite{MEF23.1}. Thus, the controller cannot dynamically change the bandwidth of the flows, as they are determined by the SLA. The Frame \textbf{Loss} Ratio (FLR) is one of the key performance attributes~\cite{MEF23.1} that a service provider must comply to, and cannot be compromised. \textbf{Splitting} a flow between two or more paths may result in packets being received out-of-order. Packet reordering is a key performance parameter in carrier-grade performance and availability measurement, as it affects various applications such as real-time media streaming~\cite{Y1563}. Thus, all packets of a flow are forwarded through the same path.

\ifdefined\TechReport
\ifdefined\ShortVersion
\textbf{Short term vs.\ long term scheduling.} 
\else
\vspace{2mm}
\emph{3) Short term vs.\ long term scheduling}
\vspace{1mm}

\fi
The OpenFlow time extension we presented in Section~\ref{DesImpSec} is intended for short term scheduling; a controller should schedule an action to a near-future time, on the order of seconds in the future.  
The challenge in long term scheduling is that during the long period between the time at which the Scheduled Bundle was sent and the time at which it is meant to be executed various external events may occur: the controller may fail or reboot, 
or a second controller\footnote{In an SDN with a distributed control plane, where more than one controller is used.} may try to perform a conflicting update.
Near future scheduling guarantees that external events that may affect the scheduled operation such as a switch reboot have a low probability of occurring.
Since near-future scheduling is on the order of seconds, this short potentially hazardous period is no worse than in conventional updates, where an OpenFlow command may be executed a few seconds after it was sent by the controller.
\fi

\ifdefined\ShortVersion
\textbf{Network latency.}
\else
\vspace{2mm}
\emph{4) Network latency}
\vspace{1mm}

\fi
In Fig.~\ref{fig:Swap}, the switches $S_1$ and $S_3$ are updated at the same time, as it is implicitly assumed that all the links have the same latency. In the general case each link has a different latency, and thus $S_1$ and $S_3$ should not be updated at the same time, but at two different times, $T_1$ and $T_3$, that account for the different latencies.

\ifdefined\TechReport

\ifdefined\ShortVersion
\textbf{Failures.} 
\else
\vspace{2mm}
\emph{5) Failures}
\vspace{1mm}

\fi
A timed update may fail to be performed in a coordinated way at multiple switches if some of the switches have failed, or if some of the controller commands have failed to reach some of the switches. Therefore, the controller uses a reliable transport protocol (TCP), in which dropped packets are retransmitted. If the controller detects that a switch has failed, or failed to receive some of the Bundle messages, the controller can use the \emph{Bundle Discard} to cancel the coordinated update. Note that the controller should send timed update messages sufficiently ahead of the scheduled time of execution, allowing enough time for possible retransmission and Discard message transmission.

\ifdefined\ShortVersion
\textbf{Controller performance overhead.} 
\else
\vspace{2mm}
\emph{6) Controller performance overhead}
\vspace{1mm}

\fi
The prototype design we presented (Fig.~\ref{fig:Arch}) uses \rptp ~\cite{ispcsrptp} to synchronize the switch and the controllers. A synchronization protocol may yield some performance overhead on the controller and switches, and some overhead on the network bandwidth. In our experiments we observed that the CPU utilization of the PTP processes in the controller in an experiment with 32 switches was $5 \%$ on the weakest machine we tested, and significantly less than $1 \%$ on the stronger machines. As for the network bandwidth overhead, accurate synchronization using PTP typically requires the controller to exchange $\sim 5$ packets per second per switch~\cite{ptpEnterprise}, a negligible overhead in high-speed networks.

\fi

\section{Conclusion}
Time and clocks are valuable tools for coordinating updates in a network. 
We have shown that dynamic traffic steering by SDN controllers requires flow swaps, which are best performed as close to instantaneously as possible. 
Time-based operation can help to achieve carrier-grade packet loss rate in environments that require rapid path reconfiguration. Our OpenFlow time extension can be used for implementing flow swaps and \timec. It can also be used for a variety of additional timed update scenarios that can help improve network performance during path and policy updates.


\ifdefined\BlindRev
\else
\section{Acknowledgments}
We gratefully acknowledge Oron Anschel and Nadav Shiloach, who implemented the \timec-enabled OFSoftswitch prototype. We thank Jean Tourrilhes and the members of the Extensibility working group of the ONF for many helpful comments that contributed to the OpenFlow time extension. We also thank Nate Foster, Laurent Vanbever, Joshua Reich and Isaac Keslassy for helpful discussions. We gratefully acknowledge the DeterLab project~\cite{DeterLabProj} for the opportunity to perform our experiments on the DeterLab testbed. This work was supported in part by the ISF grant 1520/11.
\fi



\ifdefined\TechReport
\ifdefined\JournalVer
\else
\pagebreak
\fi
\fi

\bibliographystyle{ieeetr}
\ifdefined\TechReport
\ifdefined\JournalVer
\fi
\bibliography{time}
\else
\bibliography{TimeShort}
\fi

\ifdefined\FutureVersion

\vspace{-15mm}

\begin{IEEEbiography}[{\includegraphics[width=1in,height=1.25in,clip,keepaspectratio]{./TalPhoto.jpg}}]%
{Tal Mizrahi}
is a PhD student at the Technion. He is also a switch architect at Marvell, with over 15 years of experience in networking. Tal is an active participant in the Internet Engineering Task Force (IETF), and in the IEEE~1588 working group. 
\end{IEEEbiography}

\vspace{-15mm}

\begin{IEEEbiography}[{\includegraphics[width=1in,height=1.25in,clip,keepaspectratio]{./PhotoPlaceholder}}]%
{Yoram Moses}
is the Israel Pollak academic chair and a professor of electrical engineering at the Technion. His research focuses on distributed and multi-agent systems, with a focus on fault-tolerance and on applications of knowledge and time in such systems. He is a co-author of the book Reasoning about Knowledge, recipient of the Godel prize in 1997 and the Dijkstra prize in 2009. 
\end{IEEEbiography}

\fi

\ifdefined\TechReport
\ifdefined\JournalVer
\else
\clearpage
\onecolumn
\begin{appendices}

\section{A Time Extension to the OpenFlow Protocol}
\label{ExtAppendix}
\subsection{Introduction}
This section defines a time extension to the OpenFlow protocol. This extension allows the controller to send OpenFlow commands that include an execution time, indicating to the switch \emph{when} the respective command should be performed.

As specified in~\cite{OpenFlow1.4}, a bundle is a sequence of (one or more) OpenFlow modification requests from the controller that is applied as a single OpenFlow operation. The controller uses a commit message to apply the set of requests in the bundle. Consequently, the switch applies all messages in the bundle as a single operation or returns an error. 

This extension defines \emph{scheduled bundles}; a bundle commit request may include an \emph{execution time}, specifying \emph{when} the bundle should be committed. A switch that receives a scheduled bundle, commits the bundle as close as possible to the execution time that was specified in the commit message.

This document also defines the bundle features message, allowing the controller to retrieve information about the switch's bundle support, and specifically about its scheduled bundle support.

\subsection{How It Works}
\label{HowItSec}
\vspace{2mm}
\emph{1) Overview}
\vspace{1mm}

This extension allows a bundle operation to be invoked at a scheduled time that is determined by the controller.

The time-based bundle procedure is illustrated in Figure~\ref{fig:Cont2Switch}:
\begin{enumerate}
	\item The controller starts the bundle procedure by sending an \verb|OFPBCT_OPEN_REQUEST|, and receives a reply from the switch.
	\item The controller then sends a set of $N$ \verb|OFPT_BUNDLE_ADD_MESSAGE| messages, for some $N\geq 1$.
	\item The controller MAY then send an \verb|OFPBCT_CLOSE_REQUEST|. The close request is optional, and thus the controller may skip this step.
	\item The controller sends an \verb|OFPBCT_COMMIT_REQUEST|. The \verb|OFPBCT_COMMIT_REQUEST| includes two time-related fields: the time flag and optionally the time property. When the time flag is set, it indicates that this is a \emph{scheduled commit}. A scheduled commit request includes the time property field, which contains the scheduled time at which the switch is expected to apply the bundle.
	\item After receiving the commit message, the switch applies the bundle at the scheduled time, $T_s$, and sends a \verb|OFPBCT_COMMIT_REPLY| to the controller. 
\end{enumerate}

\begin{figure}[htbp]
  \begin{center}
  \fbox{\includegraphics[width=1\textwidth]{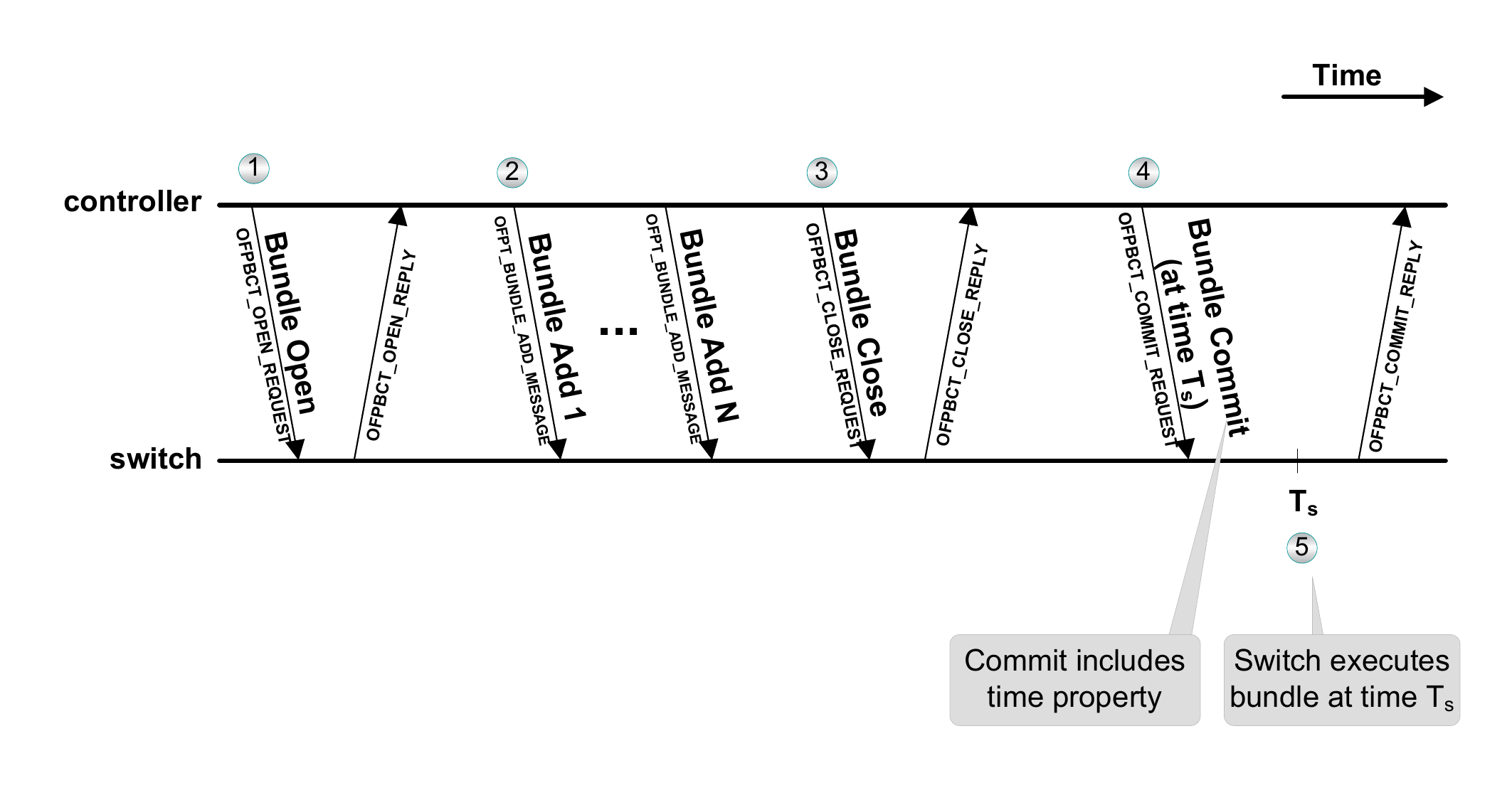}}
  \end{center}
  \caption{Scheduled Bundle Procedure}
  \label{fig:Cont2Switch}
\end{figure}

\begin{sloppypar}
\textbf{Discarding scheduled bundles.} The controller may cancel a scheduled commit by sending an \verb|OFPT_BUNDLE_CONTROL| message with type \verb|OFPBCT_DISCARD_REQUEST|. An example is shown in Figure~\ref{fig:Cont2SwitchDiscard}; if the switch is not able to schedule the operation after receiving the commit message, it responds to the controller with an error message (see~\ref{ErrorsSec}). This indication may be used for implementing a coordinated update where either all the switches successfully schedule the operation, or the bundle is discarded; when a controller receives a scheduling error message from one of the switches it can send a discard message (step 5' in in Figure~\ref{fig:Cont2SwitchDiscard}) to other switches that need to commit a bundle at the same time, and abort the bundle. 
\end{sloppypar}

\begin{figure}[htbp]
  \begin{center}
  \fbox{\includegraphics[width=1\textwidth]{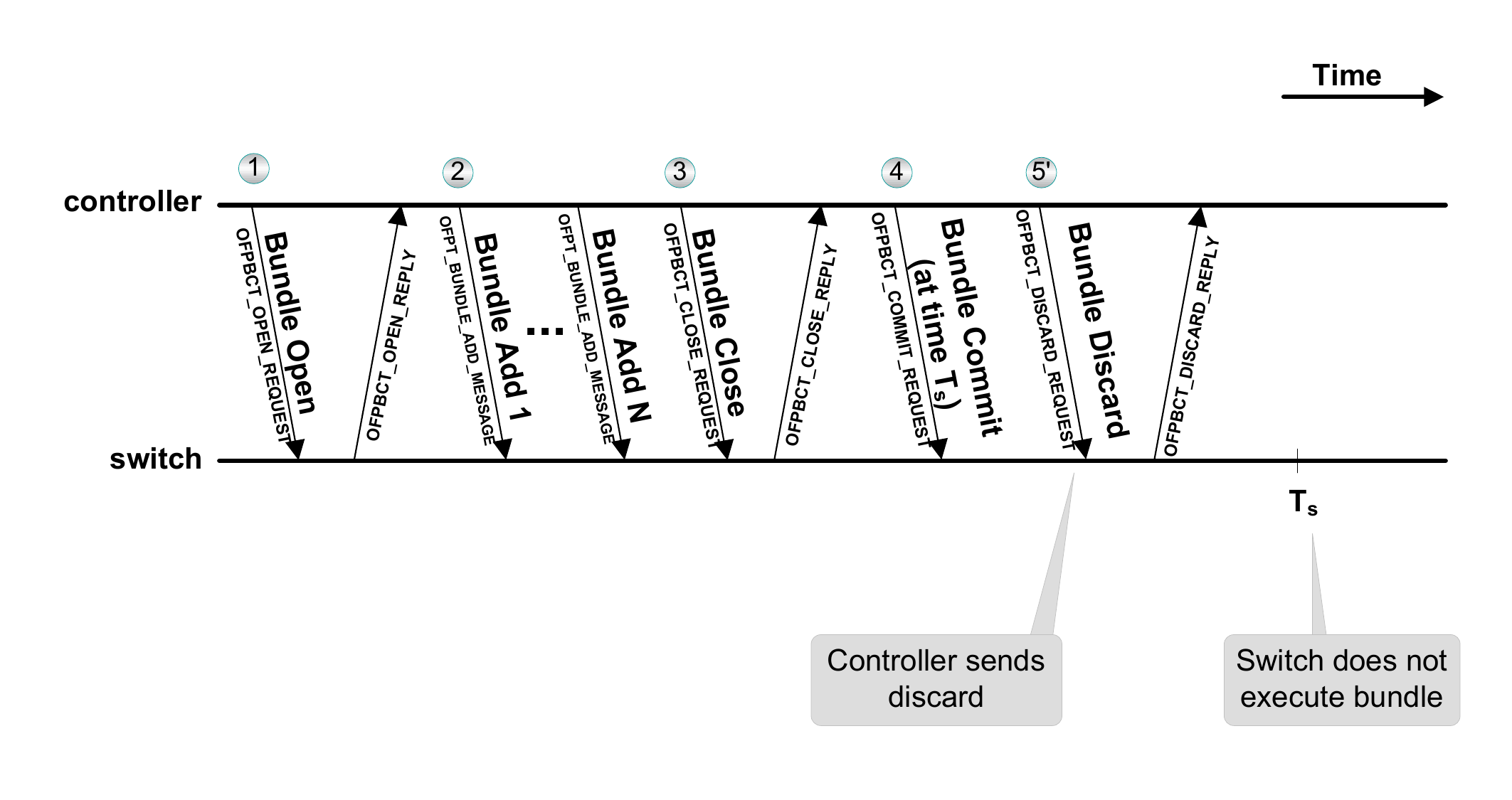}}
  \end{center}
  \caption{Discarding a Scheduled Commit}
  \label{fig:Cont2SwitchDiscard}
\end{figure}

\vspace{2mm}
\emph{2) Timekeeping and Synchronization}
\vspace{1mm}

Every switch that supports scheduled bundles must maintain a clock. It is assumed that clocks are synchronized by a method that is outside the scope of this document, e.g., the Network Time Protocol (NTP) or the Precision Time Protocol (PTP).

Two factors affect how accurately a switch can commit a scheduled bundle; one factor is the accuracy of the clock synchronization method used to synchronize the switches' clocks, and the second factor is the switch's ability to execute real-time operations, which greatly depends on how it is implemented.

This document does not define any requirements pertaining to the degree of accuracy of performing scheduled operations. However, every switch that supports the time extension is able to report its estimated scheduling accuracy to the controller. The controller can retrieve this information from the switch using the bundle features message, defined in Section~\ref{FeatureSec}.

Since a switch does not perform configuration changes instantaneously, the processing time of required operations should not be overlooked; in the context of the extension described in this paper the scheduled time and execution time always refer to the start time of the relevant operation.

\vspace{2mm}
\emph{3) Scheduling Tolerance}
\vspace{1mm}

When a switch receives a scheduled commit message, it MUST verify that the scheduled time, $T_s$, is not too far in the past or in the future. As illustrated in Figure~\ref{fig:Tolerance}, the switch verifies that $T_s$ is within the \emph{scheduling tolerance} range.

The lower bound on $T_s$ verifies the freshness of the packet so as to avoid acting upon old and possibly irrelevant messages. Similarly, the upper bound on $T_s$ guarantees that the switch does not take a long-term commitment to execute an action that may become obsolete by the time it is scheduled to be invoked. 

\begin{figure}[htbp]
  \begin{center}
  \fbox{\includegraphics[width=.5\textwidth]{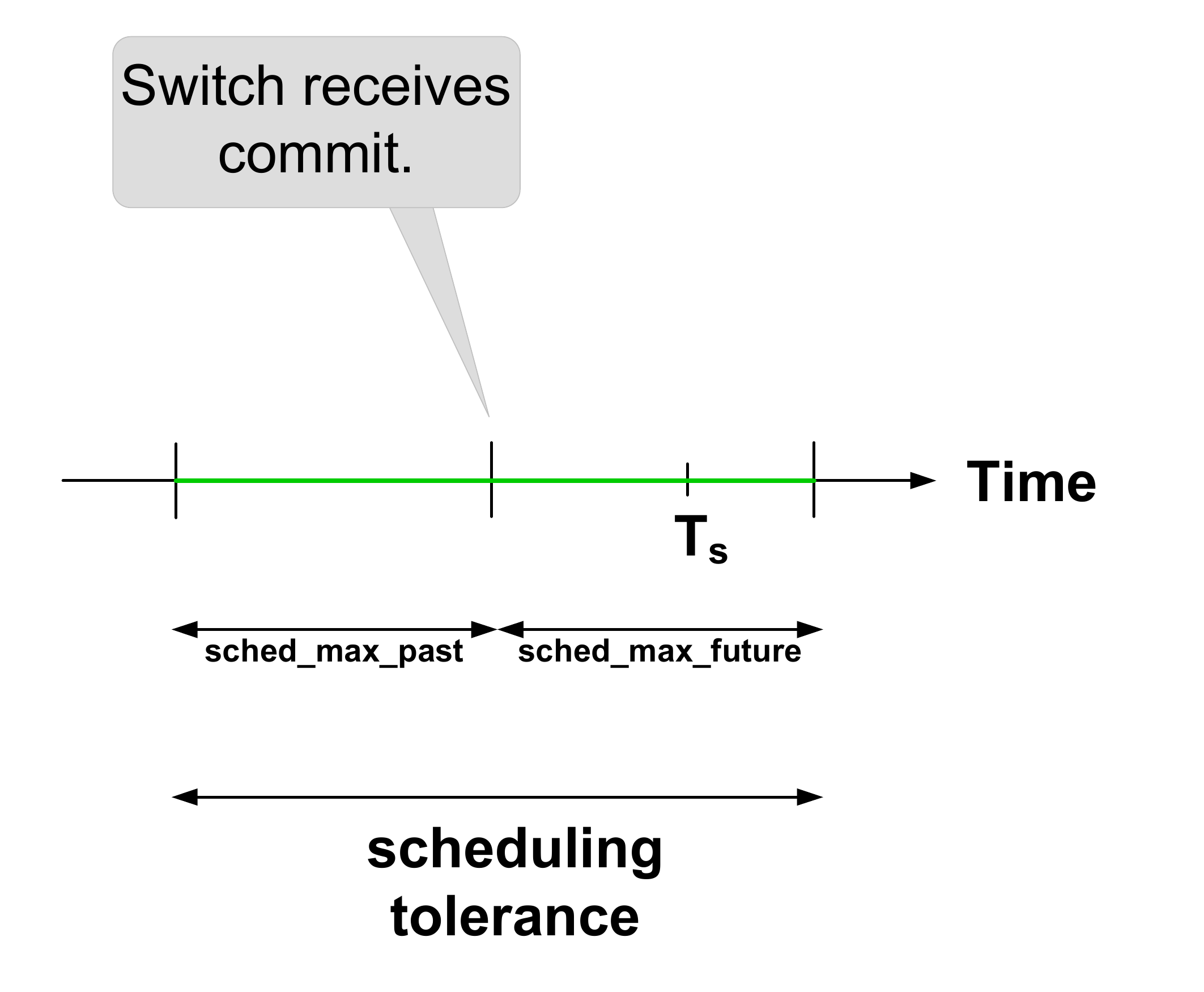}}
  \end{center}
  \caption{Scheduling Tolerance}
  \label{fig:Tolerance}
\end{figure}

The scheduling tolerance is determined by two parameters, \verb|sched_max_future| and \verb|sched_max_past|. The default value of these two parameters is 1 second. The controller MAY set these fields to a different value using the bundle features request, as described in Section~\ref{FeatureSec}.

If the scheduled time, $T_s$ is within the scheduling tolerance range, the scheduled commit is performed; if $T_s$ occurs in the past and within the scheduling tolerance, the switch applies the bundle as soon as possible. If $T_s$ is a future time, the switch applies the bundle at $T_s$. If $T_s$ is not within the scheduling tolerance range, the switch responds to the controller with an error message.

\subsection{Time-based Bundle Messages}
This section updates Section 7.3.9 of~\cite{OpenFlow1.4}. The reader is assumed to be familiar with Sections 6.8 and 7.3.9 of~\cite{OpenFlow1.4}.

The time extension allows bundle commit messages to include a time property, defining when the bundle should be executed. 

The time extension defines two time-related fields in \verb|OFPBCT_COMMIT_REQUEST| messages:
\begin{itemize}
	\item The time flag, denoted \verb|OFPBF_TIME|.
	\item The time property.
\end{itemize}

\begin{sloppypar}
All \verb|OFPT_BUNDLE_CONTROL| messages include the \verb|OFPBF_TIME| flag. In control messages with type \verb|OFPBCT_COMMIT_REQUEST| the time flag MAY be set, indicating that the time property field is present. The time property incorporates the time at which the switch is scheduled to apply the bundle.
\end{sloppypar}

Control messages with a type that is not \verb|OFPBCT_COMMIT_REQUEST| MUST have the \verb|OFPBF_TIME| flag disabled, and this flag is ignored by the switch in these messages.

\vspace{2mm}
\emph{1) The Time Flag}
\vspace{1mm}

This document updates \verb|ofp_bundle_flags| by adding the \verb|OFPBF_TIME| flag, as follows: 

\begin{footnotesize}
\begin{verbatim}
/* Bundle configuration flags. */
enum ofp_bundle_flags {
  OFPBF_TIME = 1 << 2, /* Execute in a specific time. */
};
\end{verbatim}
\end{footnotesize}

\vspace{2mm}
\emph{2) The Bundle Time Property}
\vspace{1mm}

This document defines a new bundle property, the time property. 

\begin{footnotesize}
\begin{verbatim}
/* Bundle property */
struct ofp_bundle_prop_time {
  uint16_t type;   /* OFPBPT_TIME */
  uint16_t length; /* Length in bytes = 24 */
  uint8_t pad[4];

  struct ofp_time scheduled_time;  /* The scheduled time at which the switch should apply the bundle. */
};
OFP_ASSERT(sizeof(struct ofp_bundle_prop_time) == 24);   
\end{verbatim}
\end{footnotesize}

The \verb|type| field in the time property is set to the value \verb|OFPBPT_TIME|, defined as follows:

\begin{footnotesize}
\begin{verbatim}
/* Bundle property types. */
enum ofp_bundle_prop_type {
  OFPBPT_TIME = 1, /* Time property. */
};
\end{verbatim}
\end{footnotesize}

\vspace{2mm}
\emph{3) Time Format}
\vspace{1mm}

The time format defined in this extension is based on the one defined in~\cite{IEEE1588}. It consists of two sub-fields; a \verb|seconds| field, representing the integer portion of time in seconds\footnote{The seconds field in IEEE 1588 is 48 bits long. The seconds field used in this extension is a 64-bit field, but it has the same semantics as the seconds field in the IEEE 1588 time format.}, and a \verb|nanoseconds| field, representing the fractional portion of time in nanoseconds, i.e., $0 \leq nanoseconds \leq (10^9-1)$.

\begin{footnotesize}
\begin{verbatim}
/* Time format */
struct ofp_time {
  uint64_t seconds;
  uint32_t nanoseconds;
  uint8_t pad[4];
};
OFP_ASSERT(sizeof(struct ofp_time) == 16);   
\end{verbatim}
\end{footnotesize}

As defined in~\cite{IEEE1588}, time is measured according to the International Atomic Time (TAI) timescale. The epoch is defined as 1 January 1970 00:00:00 TAI.

\subsection{Bundle Features Request}
\label{FeatureSec}

The bundle features request defined in this document allows a controller to query a switch about its bundle capabilities, including its scheduled bundle capabilities. 

This section extends Section 7.3.5 of~\cite{OpenFlow1.4}. The reader is assumed to be familiar with Section 7.3.5 of~\cite{OpenFlow1.4}.

The bundle features request is a new multipart message type, the \verb|OFPMP_BUNDLE_FEATURES| message. This document updates \verb|ofp_multipart_type| by adding the \verb|OFPMP_BUNDLE_FEATURES| type, as follows:

\begin{footnotesize}
\begin{verbatim}
enum ofp_multipart_type {
  /* Bundle features.
   * The request body is ofp_bundle_features_request.
   * The reply body is struct ofp_bundle_features. */
  OFPMP_BUNDLE_FEATURES = 17,
};
\end{verbatim}
\end{footnotesize}

\vspace{2mm}
\emph{1) Bundle Features Request Message Format}
\vspace{1mm}

The body of the bundle features request message is defined by \verb|struct ofp_bundle_features_request|, as follows:

\begin{footnotesize}
\begin{verbatim}
/* Body of OFPMP_BUNDLE_FEATURES request. */
struct ofp_bundle_features_request {
  uint32_t feature_request_flags;   /* Bitmap of "ofp_bundle_feature_flags". */
  uint8_t pad[4];
  
  /* Bundle features property list - 0 or more. */
  struct ofp_bundle_features_prop_header properties[0];
};
OFP_ASSERT(sizeof(struct ofp_bundle_features) == 8);
\end{verbatim}
\end{footnotesize}

\begin{sloppypar}
The body consists of a flags field, followed by zero or more property TLV fields. The flags field, \verb|feature_request_flags|, is defined as follows:
\end{sloppypar}

\begin{footnotesize}
\begin{verbatim}
/* Flags used in a OFPMP_BUNDLE_FEATURES request. */
enum ofp_bundle_feature_flags {
  OFPBF_TIMESTAMP = 1 << 0,      /* When enabled, the current request includes a timestamp, using 
                                  * the time property */
  OFPBF_TIME_SET_SCHED = 1 << 1, /* When enabled, the current request includes the sched_max_future  
                                  * and sched_max_past parameters, using the time property */
};
\end{verbatim}
\end{footnotesize}

If at least one of the flags \verb|OFPBF_TIMESTAMP| or \verb|OFPBF_TIME_SET_SCHED| is set, the bundle features request includes a time property.

The bundle features properties are specified below.

\vspace{2mm}
\emph{2) Bundle Features Reply Message Format}
\vspace{1mm}

If the features request is processed successfully by the switch, it sends a reply to the controller. The body of the bundle features reply message is \verb|struct ofp_bundle_features|, as follows:

\begin{footnotesize}
\begin{verbatim}
/* Body of reply to OFPMP_BUNDLE_FEATURES request. */
struct ofp_bundle_features {
  uint16_t capabilities; /* Bitmap of "ofp_bundle_flags". */
  uint8_t pad[6];

  /* Bundle features property list - 0 or more. */
  struct ofp_bundle_features_prop_header properties[0];
};
OFP_ASSERT(sizeof(struct ofp_bundle_features) == 8);
\end{verbatim}
\end{footnotesize}

\vspace{2mm}
\emph{3) Bundle Features Properties}
\vspace{1mm}

The optional property fields are defined as TLVs with a common header format, as follows:

\begin{footnotesize}
\begin{verbatim}
/* Common header for all bundle feature Properties */
struct ofp_bundle_features_prop_header {
  uint16_t type;   /* One of OFPTMPBF_*. */
  uint16_t length; /* Length in bytes of this property. */
};
OFP_ASSERT(sizeof(struct ofp_bundle_features_prop_header) == 4);
\end{verbatim}
\end{footnotesize}

The currently defined types are as follows:

\begin{footnotesize}
\begin{verbatim}
/* Bundle features property types. */
enum ofp_bundle_features_prop_type {
  OFPTMPBF_TIME_CAPABILITY = 0x1, /* Time feature property. */
  OFPTMPBF_EXPERIMENTER = 0xFFFF, /* Experimenter property. */
};
\end{verbatim}
\end{footnotesize}

\textbf{The Bundle Features Time Property.}

A bundle feature request in which at least one of the flags \verb|OFPBF_TIMESTAMP| or \verb|OFPBF_TIME_SET_SCHED| is set, incorporates the time property.
A bundle feature reply that has the \verb|OFPBF_TIME| flag set incorporates the time property.

The time property is defined as follows:

\begin{footnotesize}
\begin{verbatim}
struct ofp_bundle_features_prop_time {
  uint16_t type;   /* OFPTMPBF_TIME_CAPABILITY. */
  uint16_t length; /* Length in bytes of this property. */
  uint8_t pad[4];

  struct ofp_time sched_accuracy;   /* The scheduling accuracy, i.e., how accurately the switch can
                                     * perform a scheduled commit. This field is used only in bundle
                                     * features replies, and is ignored in bundle features requests. */
  struct ofp_time sched_max_future; /* The maximal difference between the 
                                     * scheduling time and the current time. */
  struct ofp_time sched_max_past;   /* If the scheduling time occurs in the past, defines the maximal
                                     * difference between the current time and the scheduling time. */
  struct ofp_time timestamp;        /* Indicates the time during the transmission of this message. */
};
OFP_ASSERT(sizeof(struct ofp_bundle_features_prop_time) == 72);
\end{verbatim}
\end{footnotesize}

The time property in a bundle features request includes:
\begin{sloppypar}
\begin{itemize}
	\item \verb|sched_accuracy|: this field is relevant only to bundle features replies, and the switch must ignore this field in a bundle features request.
	\item \verb|sched_max_future| and \verb|sched_max_past|: a switch that receives a bundle features request with \verb|OFPBF_TIME_SET_SCHED| set, should attempt to change its scheduling tolerance values according to the \verb|sched_max_future| and \verb|sched_max_past| values from the time property. If the switch does not successfully update its scheduling tolerance values, it replies with an error message. 
	\item \verb|timestamp|, indicating the controller's time during the transmission of this message. A switch that receives a bundle features request with \verb|OFPBF_TIMESTAMP| set, may use the received timestamp to roughly estimate the offset between its clock and the controller's clock.
\end{itemize}

The time property in a bundle features reply includes:
\begin{itemize}
	\item \verb|sched_accuracy|, indicating the estimated scheduling accuracy of the switch. For example, if the value of \verb|sched_accuracy| is $1000000$ nanoseconds (1 ms), it means that when the switch receives a bundle commit scheduled to time $T_s$, the commit will in practice be invoked at $T_s \pm 1 \ ms$. The factors that affect the scheduling accuracy are discussed in Section~\ref{HowItSec}.
	\item \verb|sched_max_future| and \verb|sched_max_past|, containing the scheduling tolerance values of the switch. If the corresponding bundle features request has the \verb|OFPBF_SET_TIME_TOLERANCE| flag enabled, these two fields are identical to the ones sent be the controller in the request.
	\item \verb|timestamp|, indicating the switch's time during the transmission of this feature reply. Every bundle feature reply that includes the time property also includes a timestamp. The timestamp may be used by the controller to get a rough estimate of whether the switch's clock is synchronized to the controller's. 
\end{itemize}
\end{sloppypar}

\subsection{Errors}
\label{ErrorsSec}
As defined in Section 7.5.4 of~\cite{OpenFlow1.4} the switch can send an error message to the controller, which includes a \verb|type| and a \verb|code|. This document extends Section 7.5.4 with additional codes, as specified below.

\vspace{2mm}
\emph{1) Bundle Error}
\vspace{1mm}

\begin{sloppypar}
When the switch has an error related to the bundle operation, it sends an error message with type \verb|OFPET_BUNDLE_FAILED|. This document defines the following new codes:
\end{sloppypar}
\begin{itemize}
	\item \verb|OFPBFC_SCHED_NOT_SUPPORTED| - this code is used when the switch does not support scheduled bundle execution and receives a commit message with the \verb|OFPBF_TIME| flag set.
	\item \verb|OFPBFC_SCHED_FUTURE| - used when the switch receives a scheduled commit message and the scheduling time exceeds the \verb|sched_max_future| (see Section~\ref{HowItSec}).
	\item \verb|OFPBFC_SCHED_PAST| - used when the switch receives a scheduled commit message and the scheduling time exceeds the \verb|sched_max_past| (see Section~\ref{HowItSec}).
\end{itemize}

The \verb|ofp_bundle_failed_code| is updated as follows:

\begin{footnotesize}
\begin{verbatim}
enum ofp_bundle_failed_code {
  OFPBFC_SCHED_NOT_SUPPORTED = 16, /* Scheduled commit was received and scheduling is not supported. */
  OFPBFC_SCHED_FUTURE = 17, /* Scheduled commit time exceeds upper bound. */
  OFPBFC_SCHED_PAST = 18, /* Scheduled commit time exceeds lower bound.  */
};
\end{verbatim}
\end{footnotesize}

\vspace{2mm}
\emph{2) Bundle Features Error}
\vspace{1mm}

When the switch has an error related to the \verb|OFPMP_BUNDLE_FEATURES| request, it replies with an error message of type \verb|OFPET_BAD_REQUEST|. The code \verb|OFPBRC_MULTIPART_BAD_SCHED| indicates that the request had the \verb|OFPBF_SET_TIME_TOLERANCE| flag enabled, and the switch failed to update the scheduling tolerance values.

The \verb|ofp_bad_request_code| is updated as follows:

\begin{footnotesize}
\begin{verbatim}
enum ofp_bad_request_code {
  OFPBRC_MULTIPART_BAD_SCHED = 16, /* Switch received a OFPMP_BUNDLE_FEATURES request and failed 
                                    * to update the scheduling tolerance. */
};
\end{verbatim}
\end{footnotesize}

\end{appendices}
\fi
\fi

\end{document}